\documentclass[12pt]{article}

\usepackage{booktabs}
\usepackage{amsmath}
\usepackage{amsthm}
\usepackage{amssymb}
\usepackage{cite}
\usepackage{url}
\usepackage[multiple]{footmisc}
\usepackage{graphicx}
\usepackage{epstopdf}
\usepackage{chngcntr}
\usepackage{enumitem}
\usepackage{hhline}
\usepackage{array}
\usepackage{hyperref}
\usepackage{color}
\usepackage{array}
\usepackage{multirow}
\usepackage{algorithm,algorithmic}
\usepackage[acronym]{glossaries}

\hypersetup{colorlinks=false}

\if@twoside
\else 
\fi

\numberwithin{equation}{section}
\counterwithin{figure}{section}
\counterwithin{table}{section}

\theoremstyle{plain}% default
\newtheorem{theorem}{Theorem}[section]

\newtheorem{proposition}[theorem]{Proposition}

\theoremstyle{definition}
\newtheorem{definition}{Definition}[section]
\newtheorem{assumption}{Assumption}[section]

\theoremstyle{remark}
\newtheorem{remark}{Remark}[section]

\newcommand{\RED}[1]{\textcolor{red}{#1}}

\definecolor{aog}{rgb}{0.0, 0.5, 0.0}

\makeatletter
\newcommand{\ALOOP}[1]{\ALC@it\algorithmicloop\ #1%
  \begin{ALC@loop}}
\newcommand{\ENDALOOP}{\end{ALC@loop}\ALC@it\algorithmicendloop}

\makeatother

\newcommand*\mystrut[1]{\vrule width0pt height0pt depth#1\relax}
\newcommand\undermat[2]{%
  \makebox[0pt][l]{$\smash{\underbrace{\phantom{%
    \begin{matrix}#2\end{matrix}}}_{\text{$#1$}}}$}#2}
\newcommand\undermatt[2]{%
  \makebox[0pt][l]{$\smash{\underbrace{\mystrut{3.5ex}\phantom{%
    \begin{matrix}#2\end{matrix}}}_{\text{$#1$}}}$}#2}    

\newacronym{doe}{DoE}{Design of Experiments}
\newacronym{glm}{GLM}{Generalized Linear Model}
\newacronym{hvof}{HVOF}{High Velocity Oxygen Fuel}
\newacronym{cfd}{CFD}{High Velocity Oxygen Fuel}
\newacronym{ann}{ANN}{Artificial Neural Network}
\newacronym{ga}{GA}{Generic Algorithm}
\newacronym{ccd}{CCD}{Central Composite Design}
\newacronym{ml}{ML}{Maximum Likelihood}
\newacronym{mle}{MLE}{Maximum Likelihood Estimation}
\newacronym{aic}{AIC}{Akaike Information Criterion}
\newacronym{loocv}{LOOCV}{Leave One Out Cross Validation}
\newacronym{pfr}{PFR}{Powder Feed Rate}
\newacronym{sod}{SOD}{Stand Off Distance}
\newacronym{cv}{CV}{Coating Velocity}
\newacronym{tgf}{TGF}{Total Gas Flow}

\makeglossaries
\title{Predictive Modelling of Critical Variables for Improving HVOF Coating using Gamma Regression Models}

\author{
Wolfgang Rannetbauer\footnote{voestalpine Stahl GmbH, voestalpine-Stra{\ss}e 3, A-4020 Linz, Austria (wolfgang.rannetbauer@voestalpine.at), Corresponding author.} ,
Simon Hubmer\footnote{Johann Radon Institute for Computational and Applied Mathematics, Altenbergerstra{\ss}e 69, A-4040 Linz, Austria, (simon.hubmer@ricam.oeaw.ac.at)} ,
\\
Carina Hambrock\footnote{voestalpine Stahl GmbH, voestalpine-Stra{\ss}e 3, A-4020 Linz, Austria (carina.hambrock@voestalpine.com)} ,
Ronny Ramlau\footnote{Johannes Kepler University Linz, Institute of Industrial Mathematics, Altenbergerstra{\ss}e 69, A-4040 Linz, Austria, (ronny.ramlau@jku.at)} \footnote{Johann Radon Institute for Computational and Applied Mathematics, Altenbergerstra{\ss}e 69, A-4040 Linz, Austria, (ronny.ramlau@ricam.oeaw.ac.at)}
}

\begin{document}

% Include the title
\maketitle

% Abstract
\begin{abstract}

Thermal spray coating is a critical process in many industries, involving the application of coatings to surfaces to enhance their functionality. This paper proposes a framework for modelling and predicting critical target variables in thermal spray coating processes, based on the application of statistical design of experiments (\acrshort{doe}) and the modelling of the data using generalized linear models (\acrshort{glm}s) with a particular emphasis on gamma regression. Experimental data obtained from thermal spray coating trials are used to validate the presented approach, demonstrating that it is able to accurately model and predict critical target variables. As such, the framework has significant potential for the optimization of thermal spray coating processes, and can contribute to the development of more efficient and effective coating technologies in various industries.

\end{abstract}

% % % % % % % % % % % % %
% Start of the sections %
% % % % % % % % % % % % %

% % % % % % % % % % % % % %
% Section - Introduction  %
% % % % % % % % % % % % % %
\section{Introduction}

Thermal spraying is a surface modification process that involves the deposition of a coating material onto a substrate by heating and accelerating a feedstock material through a spray gun. The high-velocity oxygen fuel (\acrshort{hvof}) spraying technique, schematically depicted in Figure~\ref{fig_HVOF_general}, represents a sophisticated and intricate thermal spray process that relies on the combined kinetic and thermal energy of the sprayed particles to produce coatings with exceptional properties, which makes it a subject of great interest and ongoing research in the field of materials engineering \cite{prasanna2018studies}.

Numerous techniques have been employed to model and forecast the properties of coatings produced via \acrshort{hvof} spraying. These methods include empirical models based on regression analysis, mechanistic models that simulate the physical and chemical processes occurring during spraying \cite{dongmo2008analysis, pan2016numerical, tabbara2011computational}, and hybrid models that integrate both approaches  \cite{tyagi2021evaluation}. Linear and nonlinear regression models have been used to establish a relationship between process variables and coating properties \cite{kuhnt2016residual, palanisamy2022effects, ribu2022experimental, tillmann2022statistical}, while computational fluid dynamics (\acrshort{cfd}) models have been utilized to simulate the gas flow, heat transfer, and particle behavior in the spray gun \cite{gu2001computational, li2009modeling}. In addition, artificial neural networks (\acrshort{ann}s) \cite{becker2021artificial, liu2019prediction, mojena2017neural, zhang2009characterizations} and genetic algorithms (\acrshort{ga}s) \cite{jalali2017fracture} have been implemented to optimize process conditions and predict coating properties. 

Despite the notable progress, the prediction of coating properties is still a challenging task, due to the complex interactions among the process variables, material properties, and the microstructure of the coatings. This study proposes a novel approach for modelling \acrshort{hvof} coatings through systematic variation of process variables that have received relatively limited attention in prior research. To achieve this, a Central Composite Design (\acrshort{ccd}) of experiments is employed, which enables efficient exploration of a vast parameter space. The subsequent analysis focuses on developing gamma regression models derived from generalized linear models (\acrshort{glm}s), which are particularly well-suited for modeling data with skewed, non-negative distributions. By integrating these adjusted process variables into the models, a promising opportunity arises to identify novel associations between process conditions and coating properties. 

Our approach provides insights into the intricate \acrshort{hvof} process, improving predictive models for key coating characteristics. The framework presented in this study supports the development of efficient coating technologies with enhanced attributes like wear resistance, corrosion protection, and oxidation resilience. These advancements have practical applications in industries such as aerospace, automotive, and manufacturing.

The structure of this paper is as follows: Section~\ref{sect_background} presents a comprehensive overview of the \acrshort{hvof} process, including a detailed explanation of the main factors influencing coating properties and particle in-flight characteristics. Section~\ref{sect_GLM} introduces the powerful application of generalized linear models (\acrshort{glm}s) and maximum likelihood estimation (\acrshort{mle}) to model and accurately estimate the dependence of coating properties on process conditions. This section also provides an overview of the theoretical foundations related to the asymptotic properties of \acrshort{mle} and hypothesis testing. The assessment of the predictive performance of the proposed \acrshort{glm}s is presented in Section~\ref{sect_pred}. Section~\ref{sect_application} examines the statistical design of experiments (\acrshort{doe}) and central composite design (\acrshort{ccd}) as a potent approach for efficient data collection on various levels of factors in the thermal spray coating process. The findings of the proposed framework applied to experimental data is presented in Section~\ref{sect_results}, with a specific focus on the precise prediction of critical target variables. Finally, Section~\ref{sect_conclusion} concludes the paper with a discussion on the effectiveness and potential of the proposed framework for predicting critical target variables in thermal spraying processes. It highlights the contribution of the framework to the development of more efficient coating technologies.

% % % % % % % % % % % % % % % % %
% Section - Technical Background %
% % % % % % % % % % % % % % % % % 
\section{Technical Background}\label{sect_background}

Thermal spraying is a versatile and widely used surface engineering technique that involves the deposition of coatings on the surface of a substrate to enhance its functional properties, such as wear resistance, corrosion resistance, and thermal insulation. The thermal spray coating process typically involves the application of thermal and kinetic energy to induce partial liquefaction of the coating material, thereby accelerating its projection towards the substrate surface. The amount of thermal and kinetic energy depends on the thermal spray coating technique. Various techniques, such as flame spraying, plasma spraying, arc spraying, and high-velocity oxygen fuel (\acrshort{hvof}) spraying can be used for coating using different types of coating material such as powder or wire. In this work we focus on the gas-fuel \acrshort{hvof} technology, which is described in more detail below.

\acrshort{hvof} thermal spraying has gained significant attention in recent years due to its ability to produce high-quality coatings with superior mechanical and chemical properties \cite{herman2000thermal}. The gas-fuel \acrshort{hvof} process creates its thermal energy by combustion of a mixture of oxygen and fuel gas, typically propane, methane or hydrogen, in a high-pressure chamber \cite{fauchais2014thermal}. The kinetic energy in the thermal spray process is created by its specific geometric convergent-divergent nozzle design which accelerates the gas stream to supersonic velocities. The kinetic energy is transferred to the sprayed material particles, causing them to partially melt and deform upon impact with the substrate. This results in coatings with high density and superior adhesion \cite{davis2004handbook}, almost independent of the thermal spray material composition (metallic, ceramic, cermet).

\begin{figure}[ht!]
    \centering
    \includegraphics[width=0.9\textwidth]{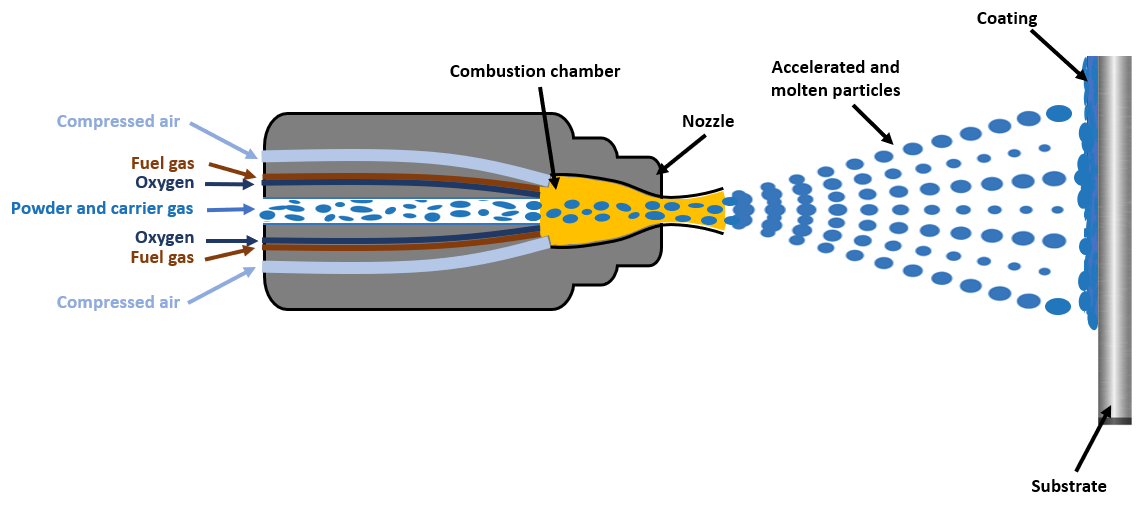}
    \caption{Schematic depiction of the High Velocity Oxygen Fuel (\acrshort{hvof}) process.}
    \label{fig_HVOF_general}
\end{figure}

Figure~\ref{fig_HVOF_general} shows a schematic depiction of the \acrshort{hvof} process, where the distinct symmetric structure of the torch becomes evident \cite{dongmo2008analysis}. The diagram reveals an axial axis of symmetry, i.e., a balanced and mirrored arrangement of components along a horizontal line within the \acrshort{hvof} system. In the combustion chamber, a defined but variable fuel gas reacts with oxygen, leading to combustion. Compressed air is used as a shroud gas and forms a cover inside the nozzle and outside in the spray plume. The powder coating material is inserted axially and is fed by using nitrogen as a carrier gas. Every characteristic of the gases and the powder, as well as the process attributes, influences combustion and consequently impacts the kinetic and thermal energy transferred to the spray particles. Influencing factors are for example:
\begin{itemize}
    \item Characteristics of gases (e.g., pressure, flow rate, temperature),
    \item Additional process variables (e.g., the ratio of the combustion gases, powder feed rate, spraying distance),
    \item Powder characteristics (e.g., size, chemical composition, density).
\end{itemize} 
The influence of these factors on the thermal and kinetic energy of the spray material can be measured by using in-situ camera equipment that is able to simultaneously detect the temperature and velocity of the sprayed particles. Process input conditions not only affect the particle characteristics but also the performance of the process, as well as the quality of the coating. Relevant measures for the performance of coating processes are
\begin{itemize}
    \item Deposition rate,
    \item Deposition efficiency.
\end{itemize} 
These performance indicators are especially important when taking into account economic aspects of the thermal spray coating process. However, the coating characteristics are the most important properties, since they influence the industrial performance. The desirable properties are, e.g.,
\begin{itemize}
    \item Specific porosity of coating,
    \item Specific hardness of coating,
    \item Specific phase and chemical composition,
    \item Specific thickness of coating.
\end{itemize} 
In order to obtain best wear resistance, corrosion resistance, or thermal insulation, a specific combination of coating properties is essential. Additionally, proper surface preparation prior to spraying is necessary to ensure maximum adhesion strength and achieve the desired coating performance characteristics. Table~\ref{tab:factors_all} provides a comprehensive overview of the various \acrshort{hvof} process variables and coating characteristics discussed above.

\begin{table}[h]
    \resizebox{\textwidth}{!}{
    \begin{tabular}{llllll}  
    \toprule
    \multicolumn{3}{c}{Controlled variables (factors)} & \multicolumn{1}{c}{In-flight-} & \multicolumn{1}{c}{Performance-} & \multicolumn{1}{c}{Coating-} \\
    \cmidrule(r){1-3}
    Gas    & Process & Powder & \multicolumn{1}{c}{properties} & \multicolumn{1}{c}{indicators} & \multicolumn{1}{c}{properties}\\
    \midrule
    pressure      & stoichiometric ratio    & size     & temperature & deposition rate & thickness\\
    flow rate          & powder feed rate        & chemical composition & velocity & deposition efficiency & roughness \\
    temperature       &  spraying distance    & density & & & phase composition \\
    &&&&&chemical composition\\
    &&&&&porosity\\
    &&&&&hardness\\
    \bottomrule
    \end{tabular}
    }
    \caption{Examples of \acrshort{hvof} process variables and characteristics.}
    \label{tab:factors_all}
\end{table}

The formulation of a robust mathematical relationship that links the input conditions controlling the spraying process, the dynamics of particles during flight, and the resulting coating characteristics is essential. Such a correlation not only facilitates a deeper understanding of the fundamental physical mechanisms underlying thermal spraying but also enables the optimization of deposition conditions to achieve the desired coating properties.

The accurate prediction of coating properties remains a challenge, primarily due to the complex and non-linear nature of the relationships between process attributes and coating properties. Achieving high accuracy in property prediction is often elusive, considering the multifaceted interactions at play. Hence, the demand for a reliable and precise prediction model becomes apparent, as it can serve as a catalyst for optimizing the process and elevating the overall quality of the resulting coatings.

% % % % % % % % % % % % % % % % % % %
% Section - Mathematical Modelling  %
% % % % % % % % % % % % % % % % % % %
\section{Predictive Modelling of \acrshort{hvof} Coating Properties}\label{sect_GLM}

The following section is dedicated to the derivation of mathematical models that enable the prediction of coating properties for the \acrshort{hvof} process. For this, we propose the use of Generalized Linear Models (\acrshort{glm}s) along with Maximum Likelihood Estimation (\acrshort{mle}) as an effective approach for modelling and estimating the expected values of target variables, conditioned on the explanatory variables (= process variables) in the \acrshort{hvof} process. We adopt the theoretical framework and notation of \cite{fahrmeir2013generalized}, to develop the statistical model used in this study, which can be expressed by the following general equation:
    \begin{equation*}
        \mathbb{E}(y_i|\boldsymbol{x_i}) = \mu_i = g^{-1}(\boldsymbol{x_i}^T\boldsymbol{\beta}),
    \end{equation*}
where $\boldsymbol{\mu}= (\mu_i)_{i=1}^n$ is a vector denoting the conditional mean of the response variable $ \mathbf{y}=(y_i)_{i=1}^n$, i.e., the coating properties of interest. The coefficient vector $\boldsymbol{\beta} = (\beta_0, \beta_1, \dots, \beta_k)$ encodes the effects of the explanatory variables ($\mathbf{x}_1, \dots, \mathbf{x}_k$), i.e., the potentially influential process conditions with observations $\mathbf{x}_i = (1, x_{1}^{i},x_{2}^{i},\dots,x_{k}^{i})_{i=1}^n$. The mean vector $\boldsymbol{\mu}$ is related to the linear combination $\boldsymbol{x_i}^T\boldsymbol{\beta}$ of the process input attributes by a one-to-one mapping $g(\cdot)$, which is often referred to as the \textit{link function} \cite{fahrmeir1985consistency}. In regression analysis, the assumption of additive random errors allows for the decomposition of the response $y_i$ into a systematic component $\mathbb{E}(y_i|\boldsymbol{x_i})$ and a random component $\epsilon_i$, yielding the equation:
    \begin{equation*}
        y_i = g^{-1}(\boldsymbol{x_i}^T\boldsymbol{\beta}) + \epsilon_i,
    \end{equation*}
where the measurement error $\epsilon_i$ is assumed to be independent of the covariates. The primary objective of regression analysis is to use the data $(y_i,\boldsymbol{x_i})_{i=1}^n$ to estimate the systematic component $\mathbb{E}(y_i|\boldsymbol{x_i})$.

Based on the broad background presented in Section \ref{sect_background}, the use of Bayesian generalized linear models to model coating properties appears to be a logical choice. The Bayesian framework offers distinct advantages over classical inference by accommodating prior knowledge regarding model parameters. This facilitates the incorporation of insights into the effects of process input variables, enhancing parameter estimation accuracy. Even in the absence of specific information, non-informative priors like the uniform distribution or Jeffreys's prior can be specified \cite{ibrahim1991bayesian}. Nevertheless, a classical frequentist statistical approach is adopted in this study to develop a model that characterizes the relationship between input variables and coating properties because the knowledge of suitable priors for Bayesian modeling is absent in this context, and the use of non-informative priors does not offer substantial advantages over classical approaches. Subsequent investigations may explore the application of Bayesian GLMs in future research.

% Subsection - Generalized Linear Models
\subsection{Generalized Linear Models (\acrshort{glm}s)}

\acrshort{glm}s, as defined in \cite{nelder1972generalized}, include a broad range of useful statistical models and serve as a powerful tool for data analysis in various fields such as engineering, physics, and biology. They extend the concept of linear regression to handle non-normal response variables, such as binary or count data, by introducing a link function that relates the mean of the response variable to the linear predictor. Here, the effectiveness of \acrshort{glm}s in analyzing data from the \acrshort{hvof} process is explored, aiming to establish a comprehensive model equation that effectively captures the conditional dependence of coating characteristics on process input variables.

In the context of \acrshort{glm}s, the response vector $ \mathbf{y} = (y_1, y_2, ..., y_n)$, is modeled as a vector that follows any distribution from the exponential family, where each element $y_i$ is distributed with a mean $\mu_i$ and variance $\sigma_i^2$. To model the relationship between the response variable $\mathbf{y}$ and the predictor variables $\mathbf{x}_i = (x_{1}^{i},x_{2}^{i},\dots,x_{k}^{i})$, a predictor vector $\boldsymbol{\eta} = (\eta_1, \eta_2, ..., \eta_n)$, with elements $\eta_i = \mathbf{x}_i^T \boldsymbol{\beta}$ is introduced, which is linked to the mean vector $\boldsymbol{\mu}= (\mu_1, \mu_2, ..., \mu_n)$ via a link function $g$, as expressed by
    \begin{equation*}
        g(\boldsymbol{\mu}) = \boldsymbol{\eta} \,.
    \end{equation*}
This formulation allows for the incorporation of multiple predictors and the estimation of their effects on the response variable. The choice of link function $g$ depends on the distribution of the response variable $\mathbf{y}$ and can vary between models. For instance, when the response variable is binary $(y_i=0$ or $y_i=1)$, the logit link function is frequently employed, connecting the mean of the response $\boldsymbol{\mu}$ variable to the logarithm of the odds ratio. This can mathematically be expressed as:
    \begin{equation*}
        g(\boldsymbol{\mu}) = \log\Big(\frac{\boldsymbol{\mu}}{1-\boldsymbol{\mu}}\Big) = \boldsymbol{\eta} \,.
    \end{equation*}
Similarly, if the response variable is a non-negative variable, the link function can be logarithmic, relating the mean of the response variable to the linear predictor via
    \begin{equation} \label{loglink}
        g(\boldsymbol{\mu}) = \log(\boldsymbol{\mu}) = \boldsymbol{\eta} \,.
    \end{equation}

After selecting the appropriate link function, the \acrshort{glm} implies that the linear predictor can be represented as a linear combination of the predictor variables $\mathbf{x}_i = (x_{1}^{i},x_{2}^{i},\dots,x_{k}^{i})$, i.e.,
    \begin{equation} \label{linearComb}
        \eta_i = \mathbf{x}_i^T \boldsymbol{\beta} =  \beta_0 + \beta_1x_{1}^{i}+ \beta_2x_{2}^{i} + \dots + \beta_kx_{k}^{i} \,,
    \end{equation}
where $\beta_0$ is called the intercept, and $\beta_1,\beta_2, \dots, \beta_k$ are the regression coefficients. The estimation of these coefficients will be performed using the method of maximum likelihood estimation, detailed in Section \ref{mle_estim}. This method quantifies the regression coefficients that are most probable given the observed data $(y_i,\boldsymbol{x_i})_{i=1}^n$, considering the assumed conditional distribution of the response variable $\mathbf{y}$ and the selected link function $g$.

% Subsection - Application of GLMs to HVOF Data
\subsection{Application of Generalized Linear Models to \acrshort{hvof} Data}

In the \acrshort{hvof} process, the response variables of interest include the coating properties, such as roughness, porosity, layer thickness, and hardness, as well as the in-flight properties, such as particle temperature and particle velocity (see Table~\ref{tab:factors_all}). Since these response variables are continuous and positive, it is necessary to choose an appropriate probability distribution and link function for modelling them. The gamma distribution with a log link function is a common choice for non-negative continuous data \cite{fahrmeir2013generalized} and thus, we will employ it in our analysis.

% Subsection - The Gamma Distribution
\subsubsection{The Gamma Distribution}

A continuous, non-negative random variable $Y$ is said to follow a gamma distribution with shape parameter $a>0$ and rate parameter $b>0$, denoted as $Y \sim G(a,b)$, if it has the density function:
    \begin{equation*}
        f(y|a,b) = \frac{b^a}{\Gamma(a)} y^{a-1}\exp(-by) \,, \qquad y > 0 \,.
    \end{equation*}
The expected value and variance are given by $\mathbb{E}(Y) = \frac{a}{b}$ and $\mathbb{V}(Y) = \frac{a}{b^2}$. An illustrative comparison of gamma distributions with varying shape and rate parameters is provided in Figure \ref{fig_pdf}. Occasionally, the gamma distribution is defined via an alternative parameterization. Depending on the expected value $\mu$ and the scale parameter $\nu > 0$, the density is then given by:
    \begin{equation*}
        f(y|\mu,\nu) = \frac{1}{\Gamma(\nu)} \Big( \frac{\nu}{\mu} \Big)^\nu\exp\Big(-\frac{\nu}{\mu}y\Big) \,,
        \qquad y > 0 \,,
    \end{equation*}
where $\mu = \mathbb{E}(Y)$ is the parameter of interest and the variance $\nu = \mathbb{V}(Y)$ is considered as a nuisance parameter, meaning that the value of $\nu$ is not the main focus of the analysis. In other words, while $\nu$ plays a role in determining the shape of the density function, it is not the parameter that one aims to estimate or draw conclusions about.

\begin{figure}[h]
    \includegraphics[width=1\textwidth]{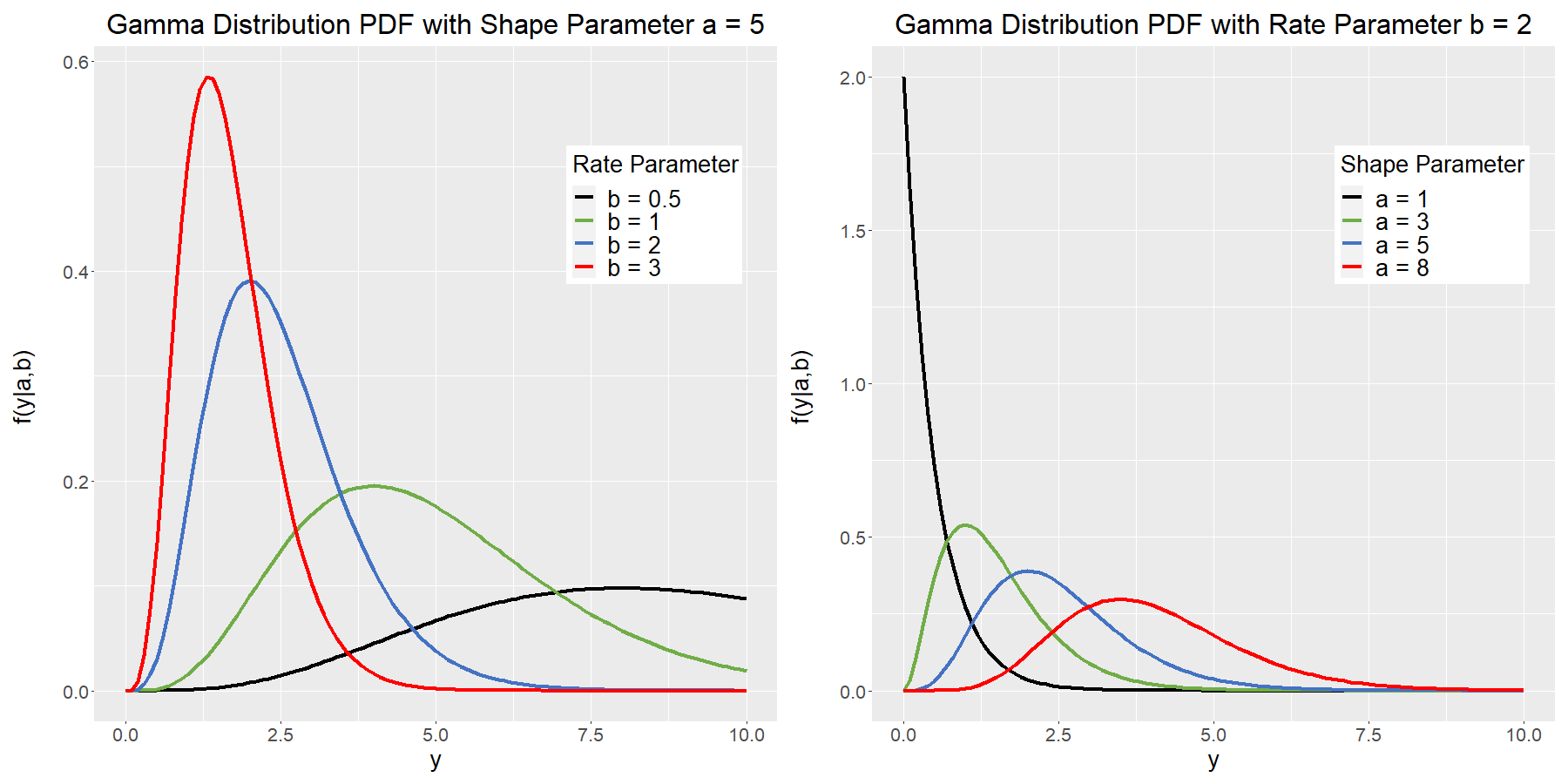}\vspace{-0.6cm}
    \caption{Comparison of gamma distributions with varied parameters. The left panel displays the probability density function (PDF) of a gamma-distributed random variable $y$ with a fixed shape parameter $a=5$ and varying rate parameters $b$. The right panel illustrates the PDF of a gamma-distributed random variable $y$ with a fixed rate parameter $b=2$ and varying shape parameters $a$.}
    \label{fig_pdf}
\end{figure}

% Subsubsection - MLE Gamma Regression
\subsubsection{Maximum Likelihood Estimation Gamma Regression}\label{mle_estim}

The likelihood function is a fundamental concept in statistical inference that quantifies the plausibility of the observed data under a given statistical model. As a consequence of the conditional independence of $y_i$ given $\boldsymbol{x_i}$, it is defined as the product of the probability density function of each observation in the sample, conditioned on the parameter values. In other words, the likelihood is the joint probability of the observed data, viewed as a function of the parameters. Therefore, the product of the likelihood contributions, defined in \ref{lik1}, yields the likelihood for the observed data, providing a basis for inference on the unknown parameters. However, it is essential to acknowledge that the assumption of conditional independence of the response variable $\mathbf{y}$ may not always hold in practical applications, particularly in complex systems or processes where various factors can influence the response variables.

In the context of this study on the coating process, assuming conditional independence implies consistent equipment- and process performance across observations, with response variable variation attributed solely to considered explanatory variables. Nonetheless, factors like equipment stability, environmental conditions, or procedural variations could introduce dependencies among response variables, challenging this assumption. To address this concern, careful execution of the experiments was performed, encompassing thorough validation of equipment functionality and accurate examination for potential issues, such as instrument cleaning and procedural consistency. The impact of changing environmental conditions can be disregarded due to the operations being conducted within a coating booth equipped with a continuous suction system. These measures were taken to validate the conditional independence assumption and ensure the reliability of the experimental results.

An empirical investigation of the dataset considered in this study reveals a notable right-skewness across all examined response variables. While marginal distributional properties alone do not necessarily dictate the choice of distributional family for modeling conditional means, this observation suggests a departure from symmetric distributional patterns. Furthermore, these variables exhibit non-negativity and continuity. To account for these specific distributional characteristics, the assumption of a gamma regression framework is made (cf. Figure \ref{fig_pdf}). In addition to the assumed gamma distribution, other distributions such as the log-normal distribution or the inverse Gaussian distribution can also be considered in this setting. Since similar results are anticipated for these alternative distributions, the emphasis on the specific distribution assumption is relaxed. Therefore, the assumption made in this study is that the response variables are conditionally gamma-distributed, with the expected value depending on the explanatory variables. Future investigations are warranted to explore the implications of alternative distributional assumptions in this domain.

The gamma regression model with a logarithmic link function assumes that the response variable $y_i$ for $i = 1, \dots, n$ follows a Gamma distribution with mean $\mu_i$ and scale parameter $\nu > 0$. The mean $\mu_i$ is modeled as a function of the covariates $\mathbf{x}_i = (x_{1}^{i},x_{2}^{i},\dots,x_{k}^{i})$ through the logarithmic transformation, which is defined as $\mu_i = \exp(\beta_0+\beta_1x_{1}^{i}+ \dots + \beta_kx_{k}^{i}) = \exp(\eta_i)$, where $\eta_i$ is the linear predictor.

For simplicity, we consider a gamma regression model with a single covariate~$\mathbf{x}$. Nevertheless, it is important to emphasize that extending the model to multiple covariates is possible and adheres to the same theoretical framework outlined here. In particular, the incorporation of additional predictors would require a simple augmentation of the linear predictor $\eta_i$ to account for their effects. Therefore, the model can be readily extended to encompass more intricate predictor configurations, as warranted by the research question at hand. The univariate regression model used in this work, in which a single dependent variable is considered, can be summarized in the form of the following equations:
    \begin{equation}\label{model}
    \begin{split}
        Y_i &\overset{\mathrm{iid}}{\sim} G(\mu_i) \quad\quad i=1,\dots,n\, , 
        \\
        \mu_i(\boldsymbol{\beta}) &= \mu_i(\boldsymbol{\beta}|x_i) = \exp(\beta_0+\beta_1x_i) = \exp(\eta_i(\boldsymbol{\beta}))\,, 
        \\
        \eta_i(\boldsymbol{\beta}) &= \eta_i(\boldsymbol{\beta}|x_i) = \beta_0+\beta_1x_i \,.
    \end{split}
    \end{equation}
The characterization of HVOF coating quality involves multiple aspects, prompting consideration of a multivariate modeling approach. Such an approach allows for the estimation of covariate effects and the assessment of a variance-covariance matrix, which quantifies correlations among different quality criteria and potentially enhances predictive capabilities. However, due to practical limitations such as the limited sample size and the need for additional data for variance-covariance matrix estimation, a univariate approach was chosen. Moreover, within the specific applications where the coated material is applied, only a subset of properties listed in Table \ref{tab:factors_all} is relevant, thus a univariate approach was considered more appropriate.

 Given the univariate model in \eqref{model} and assuming conditional independence of $y_i|\boldsymbol{x_i}$ it is now possible to define the likelihood function.
\begin{definition}
The likelihood function $L(\boldsymbol{\beta}| \mathbf{y})$ of the observed data for the model described in \eqref{model} is defined as the product of the likelihood contributions $L_i(\boldsymbol{\beta}| y_i)$, i.e.,
\begin{equation}\label{lik1}
        L(\boldsymbol{\beta}| \mathbf{y}) := \prod_{i=1}^{n} L_i(\boldsymbol{\beta}|y_i) 
        = \prod_{i=1}^{n} \frac{1}{\Gamma(\nu)} \Big( \frac{\nu}{\mu_i(\boldsymbol{\beta})} \Big)^\nu\exp\Big(-\frac{\nu}{\mu_i(\boldsymbol{\beta})}y_i\Big) \,.
    \end{equation}
\end{definition}
The log-likelihood function $\ell(\boldsymbol{\beta}| \mathbf{y}):= \log L(\boldsymbol{\beta}| \mathbf{y})$ is often preferred in statistical inference due to its numerical stability, computational simplifications, and theoretical properties \cite{fahrmeir2013generalized}.  Maximizing $\ell(\boldsymbol{\beta}| \mathbf{y})$, instead of \eqref{lik1} enables accurate parameter estimation and reliable statistical inference.

In the process of maximizing the log-likelihood function, the score function serves as an essential mathematical tool. It accurately measures the sensitivity of the log-likelihood to changes in the parameters of interest.

\begin{definition}\label{score}
The score function $\boldsymbol{s({\beta}| \mathbf{y})}$ is defined as the gradient of the log-likelihood function with respect to the model parameters, i.e., 
    \begin{equation*}
        \boldsymbol{s({\beta}| \mathbf{y})} := \nabla_{\beta} \ \ell(\boldsymbol{\beta}| \mathbf{y})\,.
    \end{equation*}
Furthermore, the maximum likelihood estimator (\acrshort{mle}) $\boldsymbol{\hat{\beta}}$ is defined as the solution of 
    \begin{equation*}
        \boldsymbol{s(\hat{\beta}| \mathbf{y})} = \mathbf{0} \,.
    \end{equation*}
\end{definition}

The score function $\boldsymbol{s({\beta}| \mathbf{y})}$ quantifies the rate of change of the log-likelihood function as the parameter values are varied, and provides a measure of the direction and magnitude of the parameter updates that increase the log-likelihood. It can be computed by either numerically or analytically differentiating $\ell(\boldsymbol{\beta}| \mathbf{y})$, which in our case leads to 

\begin{proposition} \label{prop1}
Let $\boldsymbol{s({\beta}| \mathbf{y})}$ represent the score function for a response vector $\mathbf{y}$, as defined in Definition \ref{score}, and let $\mathbf{y}$ consist of observed values $y_i$ from a random variable $Y_i$ following a gamma distribution, as described in \eqref{model}. Then
    \begin{equation*}
        \boldsymbol{s(\beta| \mathbf{y})} = \mathbf{X}^T \ \nu \ \Big( \frac{\mathbf{y}}{\boldsymbol{\mu}(\boldsymbol{\beta})} - 1 \Big) \,.
    \end{equation*}
Here, $\mathbf{X} = (\mathbf{1}\ \mathbf{x})$ represents the design matrix, consisting of the explanatory variable $\mathbf{x}~= (x_1, \dots, x_n)^T$, where $x_i \in \mathbb{R}$, $\mathbf{y} = (y_1, \dots, y_n)^T$ is the response vector with $y_i \in \mathbb{R}^+$, $\boldsymbol{\mu}(\boldsymbol{\beta}) = (\mu_1(\boldsymbol{\beta}), \dots, \mu_n(\boldsymbol{\beta}))^T$ is the mean vector with $\mu_i(\boldsymbol{\beta}) \in \mathbb{R}^+$, given in \eqref{model}, and $\mathbf{1} = (1, \dots, 1)^T$ is a vector of ones. 
\end{proposition}

\begin{proof} 
This proof is adapted from \cite{fahrmeir2013generalized}, with suitable changes accounting for the gamma regression framework considered here. First of all, the first partial derivatives of individual log-likelihoods $\log L_i(\boldsymbol{\beta}|y_i)$ are given by
    \begin{equation*}
    \begin{split}
        \frac{\partial \log L_i(\beta_0,\beta_1|y_i)}{\partial \beta_0} &= \Big( -\frac{\nu}{\mu_i(\boldsymbol{\beta})} + \frac{\nu}{\mu_{i}(\boldsymbol{\beta})^{2}} y_i \Big) \ \mu_i(\boldsymbol{\beta}) \quad
        = \Big( \frac{\nu}{\mu_i(\boldsymbol{\beta})} y_i - \nu \Big) \ \quad \\
        &= \nu\ \Big( \frac{y_i}{\mu_i(\boldsymbol{\beta})} - 1 \Big),  
        \\
        \frac{\partial \log L_i(\beta_0,\beta_1|y_i)}{\partial \beta_1} &= \Big( -\frac{\nu}{\mu_i(\boldsymbol{\beta})} + \frac{\nu}{\mu_{i}(\boldsymbol{\beta})^{2}} y_i \Big) \ \mu_i(\boldsymbol{\beta}) \ x_i 
        = \Big( \frac{\nu}{\mu_i(\boldsymbol{\beta})} y_i - \nu \Big) \ x_i \\
        &= \nu\ x_i \ \Big( \frac{y_i}{\mu_i(\boldsymbol{\beta})} - 1 \Big) \,,
    \end{split}
    \end{equation*}
Together with the definitions of the vectors $\mathbf{x}$, $\mathbf{y}$, $\boldsymbol{\mu}(\boldsymbol{\beta})$, and $\mathbf{1}$, as well as of the design matrix $\mathbf{X} = (\mathbf{1}\ \mathbf{x})$ and the definition of the score function $\boldsymbol{s(\beta|\mathbf{y})}$ there holds
    \begin{equation*}
    \begin{split}
        \mathbf{s}(\beta_0,\beta_1|\mathbf{y}) &= \begin{pmatrix} 
	    \sum_{i=1}^{n} \nu \Big( \frac{y_i}{\mu_i(\boldsymbol{\beta})} -1 \Big) 
        \\
        \sum_{i=1}^{n} \nu \ x_i \Big( \frac{y_i}{\mu_i(\boldsymbol{\beta})} -1 \Big) 
        \end{pmatrix}
        =\begin{pmatrix} 
        \mathbf{1}^T \nu \Big( \frac{\mathbf{y}}{\boldsymbol{\mu}(\boldsymbol{\beta})} - 1 \Big) 
        \\
        \mathbf{x}^T \nu \Big( \frac{\mathbf{y}}{\boldsymbol{\mu}(\boldsymbol{\beta})} - 1 \Big) 
        \end{pmatrix} 
        = \mathbf{X}^T \ \nu \ \Big( \frac{\mathbf{y}}{\boldsymbol{\mu}(\boldsymbol{\beta})} - 1 \Big) \,,
    \end{split}
    \end{equation*}
which completes the proof. 
\qed
\end{proof}

\begin{remark}
 Note that the specific form of the design matrix $\mathbf{X}$, as described here, applies only to the model under consideration. In general, the design matrix $\mathbf{X}$ comprises not only the original explanatory variables but also their higher-order powers and/or products. This expanded form allows for the estimation of higher-order effects or interaction effects between two or more covariates. \cite{fahrmeir2013generalized}
\end{remark}

By setting the score function $\boldsymbol{s(\beta|\mathbf{y})}$ to zero, a linear system of equations for $(\beta_0,\beta_1)$ arises that needs to be solved numerically. The numerical algorithm used in this work (cf. Section~\ref{methods}) involves the computation of the observed information matrix $\boldsymbol{H(\beta|\mathbf{y})}$ (= Hessian matrix) or expected information matrix $\boldsymbol{F(\beta|\mathbf{y})}$ (= Fisher matrix), which is a key component of the algorithm. Note that setting the score function $\boldsymbol{s(\beta|\mathbf{y})}$ to zero is independent of $\nu$, meaning that the process of finding solutions for $(\beta_0, \beta_1)$ is not influenced by the value of $\nu$. While solving the equation $\boldsymbol{s(\beta|\mathbf{y})} = 0$ does not inherently guarantee a maximum, the concave nature of the log-likelihood function $\ell(\boldsymbol{\beta}| \mathbf{y})$ in this model ensures its maximization \cite{wedderburn1976existence}. This concavity can be confirmed by verifying the positive semi-definiteness of $\boldsymbol{H(\beta|\mathbf{y})}$, defined in \ref{H}. For further insights into the existence and uniqueness of the maximum likelihood estimator in generalized linear models, refer to \cite{wedderburn1976existence}.

In general, there is no guarantee that solving the equation function $\boldsymbol{s(\beta|\mathbf{y})} = 0$ yields a maximum. However, since the log-likelihood $\ell(\boldsymbol{\beta}| \mathbf{y})$ for this model is a concave function, solving the equation will yield a maximum. The concativity for $\ell(\boldsymbol{\beta}| \mathbf{y})$, defined in \ref{H}, can be verified by checking if $\boldsymbol{H(\beta|\mathbf{y})}$ is positive definite. 

\begin{definition}
The observed information matrix $\boldsymbol{H(\beta|\mathbf{y})}$ is defined as the Hessian matrix of the log-likelihood function $\ell(\boldsymbol{\beta|\mathbf{y}})$, i.e., the matrix of second derivatives with respect to the model parameters $\boldsymbol{\beta}$,
    \begin{equation}\label{H}
        \boldsymbol{H(\beta|\mathbf{y})} := - \frac{\partial^2 \ell(\beta_0,\beta_1|\mathbf{y})}{\partial \boldsymbol{\beta} \ \partial \boldsymbol{\beta}^T} \,.       
    \end{equation}
The expected information matrix $\boldsymbol{F(\beta|\mathbf{y})}$ is defined as 
    \begin{equation}\label{F}
        \boldsymbol{F(\beta|\mathbf{y})} := \mathbb{E} \Big[- \frac{\partial^2 \ell(\beta_0,\beta_1|\mathbf{y})}{\partial \boldsymbol{\beta} \ \partial \boldsymbol{\beta}^T}\Big] \,,  
    \end{equation}
where $\mathbb{E}[\cdot]$ denotes the expected value.
\end{definition}

The matrices $\boldsymbol{H(\beta|\mathbf{y})}$ and $\boldsymbol{F(\beta|\mathbf{y})}$ quantify the amount of information that the observed data provides about the unknown parameters of the model. For our specific gamma regression framework, they can be computed explicitly as described in the following

\begin{proposition}
Let $\mathbf{y}$ be as in Proposition~\ref{prop1}, let $\mathbf{W}$ = $diag( \RED{\nu ~ y_i} / \mu_i(\boldsymbol{\beta}) )_{i=1,\dots,n}$ be a diagonal matrix with elements $\nu \ y_i/\mu_i(\boldsymbol{\beta})$ and $\mathbf{\Tilde{W}}$ = $diag(\nu)$ be a diagonal matrix with elements ${\nu}$. Then the observed information matrix $\boldsymbol{H(\beta|\mathbf{y})}$ and the expected information matrix $\boldsymbol{F(\beta|\mathbf{y})}$, defined in \eqref{H} and \eqref{F}, respectively, can be expressed as
    \begin{equation*}
        \boldsymbol{H(\beta|\mathbf{y})} = \mathbf{X}^T\mathbf{W}\mathbf{X} \,,
        \qquad \text{and} \qquad
        \boldsymbol{F(\beta|\mathbf{y})} = \mathbf{X}^T\mathbf{\Tilde{W}}\mathbf{X} \,.
    \end{equation*}
\end{proposition}
\begin{proof}
The second partial derivatives of individual log-likelihoods $\log L_i(\boldsymbol{\beta|\mathbf{y}})$ are given by
    \begin{equation}\label{decondP}
    \begin{split}
        \frac{\partial^2 \log L_i(\beta_0,\beta_1|\mathbf{y})}{\partial \beta_0^2} 
        &= -\frac{\nu\ y_i}{\mu_i(\boldsymbol{\beta})} \,, \quad \quad
        \frac{\partial^2 \log L_i(\beta_0,\beta_1|\mathbf{y})}{\partial \beta_1^2} = -\frac{\nu\ x_i^2 \ y_i}{\mu_i(\boldsymbol{\beta})} \,, 
        \\
        \frac{\partial^2 \log L_i(\beta_0,\beta_1|\mathbf{y})}{\partial \beta_0 \ \beta_1} &= -\frac{\nu\ x_i \ y_i}{\mu_i(\boldsymbol{\beta})} \,.
    \end{split}
    \end{equation}
The observed information matrix $\boldsymbol{H(\beta|\mathbf{y})}$ is obtained through the aggregation of the second partial derivatives of individual log-likelihoods $\log L_i(\boldsymbol{\beta|\mathbf{y}})$, i.e.,
    \begin{equation*}
    \begin{split}
        \boldsymbol{H(\beta|\mathbf{y})} &= - \frac{\partial^2 \ell(\beta_0,\beta_1|\mathbf{y})}{\partial \boldsymbol{\beta} \ \partial \boldsymbol{\beta}^T} \overset{\eqref{lik1}}{=} - \sum_{i=1}^{n}  \frac{\partial^2 \log L_i(\beta_0,\beta_1|\mathbf{y})}{\partial \boldsymbol{\beta} \ \partial \boldsymbol{\beta}^T} 
        \overset{\eqref{decondP}}{=} \begin{pmatrix} 
            \sum_{i=1}^{n} \frac{\nu \ y_i}{\mu_i(\boldsymbol{\beta})} & \sum_{i=1}^{n} \frac{\nu \ x_i \ y_i}{\mu_i(\boldsymbol{\beta})} \\
            \sum_{i=1}^{n} \frac{\nu \ x_i \ y_i}{\mu_i(\boldsymbol{\beta})} &  \sum_{i=1}^{n} \frac{\nu \ x_i^2 \ y_i}{\mu_i(\boldsymbol{\beta})}
        \end{pmatrix} \,.
    \end{split}
    \end{equation*}
Together with the definition of $\mathbf{W}$ we thus obtain
    \begin{equation*}
        \boldsymbol{H(\beta|\mathbf{y})} = \mathbf{X}^T\mathbf{W}\mathbf{X} \,.
    \end{equation*}
Since $\mathbb{E}(y_i) = \mu_i(\boldsymbol{\beta})$, the Fisher matrix $\boldsymbol{F(\beta|\mathbf{y})}$ is given by
    \begin{equation*}
    \begin{split}
        \boldsymbol{F(\beta|\mathbf{y})} = \mathbb{E} \Big[- \frac{\partial^2 \ell(\beta_0,\beta_1|\mathbf{y})}{\partial \boldsymbol{\beta} \ \partial \boldsymbol{\beta}^T}\Big] &= \begin{pmatrix} 
        \sum_{i=1}^{n} \frac{\nu \ \mathbb{E}(y_i)}{\mu_i(\boldsymbol{\beta})} & \sum_{i=1}^{n} \frac{\nu \ x_i \ \mathbb{E}(y_i)}{\mu_i(\boldsymbol{\beta})} \\ \sum_{i=1}^{n} \frac{\nu \ x_i \ \mathbb{E}(y_i)}{\mu_i(\boldsymbol{\beta})} &  \sum_{i=1}^{n} \frac{\nu \ x_i^2 \ \mathbb{E}(y_i)}{\mu_i(\boldsymbol{\beta})}
        \end{pmatrix} 
        \\ 
        &= \begin{pmatrix} 
        \sum_{i=1}^{n} \frac{\nu \ \mu_i(\boldsymbol{\beta})}{\mu_i(\boldsymbol{\beta})} & \sum_{i=1}^{n} \frac{\nu \ x_i \ \mu_i(\boldsymbol{\beta})}{\mu_i(\boldsymbol{\beta})} \\
        \sum_{i=1}^{n} \frac{\nu \ x_i \ \mu_i(\boldsymbol{\beta})}{\mu_i(\boldsymbol{\beta})} &  \sum_{i=1}^{n} \frac{\nu \ x_i^2 \ \mu_i(\boldsymbol{\beta})}{\mu_i(\boldsymbol{\beta})}
        \end{pmatrix} =
        \mathbf{X}^T\mathbf{\Tilde{W}}\mathbf{X} \,,
    \end{split}
    \end{equation*}
which yields the assertion.
\qed
\end{proof}

% Subsubsection - Numerical Computation of the MLE
\subsubsection{Numerical Computation of the Maximum Likelihood Estimator} \label{methods}

Numerical algorithms are essential for estimating the maximum likelihood estimator of parameters in a statistical model, particularly when an analytical solution to the likelihood equations is unattainable \cite{hardin2007generalized}. Frequently, the likelihood function is an intricate, nonlinear function of parameters, lacking a closed-form expression for its maximum, e.g.,  in gamma regression with a logarithmic link function.

In such cases, numerical algorithms such as the Newton-Raphson algorithm are employed to iteratively approximate the solution of the likelihood equations until convergence is reached \cite{fahrmeir2013generalized}. These methods rely on numerical techniques to estimate the derivatives of the likelihood function, which are used in computing the updates to the parameter estimates.

\begin{description}
\item[Newton-Raphson Method \cite{hardin2007generalized}] is an iterative method used to find a value of $\boldsymbol{\beta}$ that satisfies the equation $\mathbf{s}(\boldsymbol{\beta|\mathbf{y}}) = 0$, which corresponds to the point where the log-likelihood function is maximized. The Newton-Raphson algorithm achieves this by iteratively approximating the solution of $\mathbf{s}(\boldsymbol{\beta|\mathbf{y}}) = 0$ using Taylor series expansion of $\mathbf{s}(\boldsymbol{\beta|\mathbf{y}})$ around the current estimate of $\boldsymbol{\beta}$. Specifically, the expansion can be written as:
    \begin{equation}\label{tangent}
        \mathbf{s}(\boldsymbol{\beta|\mathbf{y}}) \approx \mathbf{s}(\boldsymbol{\beta}^{(k)}|\mathbf{y}) - \boldsymbol{H(\beta}^{(k)}|\mathbf{y}) (\boldsymbol{\beta} - \boldsymbol{\beta}^{(k)}) \,,
    \end{equation}
where $\boldsymbol{\beta}^{(k)}$ is the estimate of $\boldsymbol{\beta}$ at the $k$-th iteration, $\mathbf{s}(\boldsymbol{\beta}^{(k)}|\mathbf{y})$ is the score function evaluated at $\boldsymbol{\beta}^{(k)}$, and $\boldsymbol{H(\beta}^{(k)}|\mathbf{y}) = - \partial \ \mathbf{s}(\boldsymbol{\beta}^{(k)}|\mathbf{y}) / \partial \ \boldsymbol{\beta}^T $ is the observed information matrix evaluated at $\boldsymbol{\beta}^{(k)}$. The score function is approximated using a linear tangent line, resulting in an improved approximate solution. This involves finding the root of the tangent line in \eqref{tangent}. Thus, the algorithm approximates the maximum likelihood estimator of $\boldsymbol{\beta}$ by solving the equation:
    \begin{equation*} 
        \mathbf{s}(\boldsymbol{\beta}^{(k)}|\mathbf{y}) - \boldsymbol{H(\beta}^{(k)}|\mathbf{y}) (\boldsymbol{\beta} - \boldsymbol{\beta}^{(k)}) = 0 \,,
    \end{equation*}
for $\boldsymbol{\beta}$, which yields:
    \begin{equation}\label{update}
        \boldsymbol{\beta}^{(k+1)} = \boldsymbol{\beta}^{(k)} + \boldsymbol{H(\beta}^{(k)}|\mathbf{y})^{\dagger} \ \mathbf{s}(\boldsymbol{\beta}^{(k)}|\mathbf{y}) \,.
    \end{equation}
    The algorithm iterates until convergence is achieved, which is typically defined as the point at which the change in the estimate of $\boldsymbol{\beta}$ between two successive iterations falls below a certain threshold.\\

\item[Fisher Scoring Method \cite{fahrmeir2013generalized}] is a useful approach for maximum likelihood estimation that involves replacing the observed information matrix $\boldsymbol{H(\beta}^{(k)}|\mathbf{y})$ by the expected information matrix $\boldsymbol{F(\beta}^{(k)}|\mathbf{y})$ in the update formula \eqref{update}, i.e., 
    \begin{equation}\label{updateF}
        \boldsymbol{\beta}^{(k+1)} = \boldsymbol{\beta}^{(k)} + \boldsymbol{F(\beta}^{(k)}|\mathbf{y})^{\dagger} \ \mathbf{s}(\boldsymbol{\beta}^{(k)}|\mathbf{y}) \,.
    \end{equation}
    This simplifies the required computations, making it faster and more efficient. 
\end{description}

% Subsubsection - Asympototic Properties of the MLE
\subsubsection{Asymptotic Properties of the Maximum Likelihood Estimator (\acrshort{mle})}

Given the gamma regression model with logarithmic link function, as defined in \eqref{model}, and the \acrshort{mle} procedure presented in the previous section, we now investigate the asymptotic properties of the \acrshort{mle} of the regression coefficients $\boldsymbol{\beta} = (\beta_0, \dots, \beta_k)^T$. Specifically, under mild regularity conditions introduced below, the \acrshort{mle} can be proven to be a consistent and asymptotically normal estimator, with its asymptotic covariance matrix being equivalent to the inverse of the Fisher information matrix \cite{fahrmeir1985consistency}.

\begin{assumption}[\cite{fahrmeir1985consistency} Regularity Assumptions]\label{reg}
Let $\hat{\boldsymbol{\beta}} \in B \subset \mathbb{R}^p$ denote the \acrshort{ml} estimator for the true parameter $\boldsymbol{\beta}$, $p$ be the number of predictor variables in the model, and $M$ the image $\boldsymbol{\mu(\beta)}$ of $\boldsymbol{\beta}$. Furthermore, the linear combination of the predictor variables $\boldsymbol{\eta}$ is related to the mean $\boldsymbol{\mu(\beta)}$ of the response $y$ by an injective link function $g: M \rightarrow \mathbb{R}^p$, i.e., $\boldsymbol{\eta} = g(\boldsymbol{\mu(\beta)})$ (compare with \eqref{model}). Additionally, there holds
\begin{itemize}
    \item[(i)] $B$ is open in $\mathbb{R}^p$,
    \item[(ii)] The design matrix $\mathbf{X}$ has full rank, i.e., $rank(\mathbf{X})=p$, 
    \item[(iii)] $g(\cdot)$ is twice continuously differentiable on M.
\end{itemize}
\end{assumption}

Note that Assumption~\ref{reg} is valid for our gamma regression model with a logarithmic link function \eqref{loglink}, i.e., where the response variable follows a gamma distribution \eqref{model}.

\begin{definition}
An estimator $\hat{\boldsymbol{\beta}}$ is consistent for the true parameter vector $\boldsymbol{\beta}$ if, as the sample size $n$ goes to infinity, $\hat{\boldsymbol{\beta}}$ converges in probability to $\boldsymbol{\beta}$. In other words, for any small positive number $\epsilon$, it holds that 
    \begin{equation*}
        \lim_{n \to \infty} P(||\hat{\boldsymbol{\beta}} - \boldsymbol{\beta}|| > \epsilon) = 0 \,.
    \end{equation*}
\end{definition}

Using the Law of Large Numbers, the sample mean of a sequence of i.i.d.\ random variables with finite mean converges in probability to the expected value. Since the log-likelihood function $\ell(\boldsymbol{\beta}| \mathbf{y})$ in this model \eqref{model} is the sum of i.i.d.\ Gamma distributions, the Law of Large Numbers can be used to establish convergence in probability of the \acrshort{mle} to the true parameter values. In the following proposition, two key properties of the gamma regression model are established without providing formal proof.

\begin{proposition}[\cite{fahrmeir1985consistency}] \label{consistency} \hfill
\begin{enumerate}[label=(\roman*)]
         \item In the setting of the gamma regression model \eqref{model}, the \acrshort{mle} $\hat{\boldsymbol{\beta}}$ is consistent for $\boldsymbol{\beta}$. In particular, under the regularity conditions stated in Assumption~\ref{reg}, the \acrshort{ml} estimator $\hat{\boldsymbol{\beta}}$ converges in probability to the true regression coefficients $\boldsymbol{\beta}$ for increasing sample size, i.e., $\hat{\boldsymbol{\beta}} \overset{p}{\to} \boldsymbol{\beta}$, where $\overset{p}{\to}$ denotes convergence in probability.
        \item Let the assumptions of Proposition~\ref{consistency} hold. Then the gamma regression model defined in \eqref{model} is asymptotically normal in relation to the maximum likelihood estimator (\acrshort{mle}) $\hat{\boldsymbol{\beta}}$, i.e., $\sqrt{n}(\hat{\boldsymbol{\beta}} - \boldsymbol{\beta}) \overset{d}{\to} \mathcal{N}(\boldsymbol{0}, \boldsymbol{F}^{\dagger}(\boldsymbol{\beta|\mathbf{y}}))$, where $\overset{d}{\to}$ denotes convergence in distribution, and $n$ denotes the sample size.
\end{enumerate}
\end{proposition}

% Subsubsection - Linear Hypothesis Testing
\subsubsection{Linear Hypothesis Testing}
\label{testing}

By conducting hypothesis tests on the estimated regression coefficients $\hat{\boldsymbol{\beta}}$, one can provide evidence-based justifications for the inclusion or exclusion of specific predictors, ensure the robustness and reliability of a model, and enhance the interpretability and generalizability of the findings. Testing a linear hypothesis on the coefficients of the Gamma \acrshort{glm} can be represented as follows:
    \begin{equation*}
        H_0: \boldsymbol{C}\boldsymbol{\beta} = \boldsymbol{d} \,,
    \end{equation*}
where $\boldsymbol{C}$ is a known $r\times p$ matrix of rank $r$, $\boldsymbol{\beta}$ is the $p\times 1$ vector of regression coefficients, and $\boldsymbol{d}$ is the $r\times 1$ vector of known constants. This matrix $\boldsymbol{C}$ is used to define the specific hypothesis being tested, and its structure depends on the research question at hand. In the context of our study, $\boldsymbol{C}$ is constructed to examine the significance of certain predictors in relation to the response variable.

Under hypothesis $H_0$, the unrestricted maximum likelihood estimator $\hat{\boldsymbol{\beta}}$ is not efficient, and therefore we need to consider restricted estimators that take into account the constraints imposed by $H_0$ \cite{fahrmeir2013generalized}. For this, we consider the Wald statistic $w$ given in

\begin{definition} \label{defWald}
The Wald statistic $w$ is defined as:
    \begin{equation}\label{wald}
        w = (\boldsymbol{C}\hat{\boldsymbol{\beta}}-\boldsymbol{d})^T\left[\boldsymbol{C} \underbrace{(\boldsymbol{X}^T\boldsymbol{\Tilde{W}}\boldsymbol{X})
        \ \hspace{-0.2cm} ^{\dagger}}_{\boldsymbol{F^\dagger(\hat{\beta}|\mathbf{y})}}   \boldsymbol{C}^T\right]^{-1}(\boldsymbol{C}\hat{\boldsymbol{\beta}}-\boldsymbol{d}) \,,
    \end{equation}
where $\boldsymbol{X}$ is the $n\times p$ design matrix, and $\boldsymbol{W}$ is the $n\times n$ diagonal matrix with the weights $w_i$ on the diagonal. 
\end{definition}

Under $H_0$, the Wald statistic has an asymptotic $\chi^2$-distribution with $r$ degrees of freedom \cite{fahrmeir2013generalized}, i.e.,
    \begin{equation*}
        w \stackrel{d}{\rightarrow} \chi_r^2 \quad \text{as } n \rightarrow \infty \,.
    \end{equation*}
We reject $H_0$ at level $\alpha$ if $w > \chi_{r,1-\alpha}^2$, where $\chi_{r,1-\alpha}^2$ is the $1-\alpha$ quantile of the $\chi^2$-distribution with $r$ degrees of freedom.

In the specific case of predictive modelling in \acrshort{hvof} coating, hypothesis testing plays a vital role in determining the relevance of regression coefficients $\boldsymbol{\beta}_j$, where $\boldsymbol{\beta}_j$ denotes a subvector of $\boldsymbol{\beta}$. Specifically, we consider the case where the null hypothesis $H_0: \boldsymbol{\beta}_j = 0$ versus the alternative hypothesis $H_1: \boldsymbol{\beta}_j \neq 0$.

\begin{proposition} \label{prop3.5}
Let $\boldsymbol{\beta}_j$ be a subvector of $\boldsymbol{\beta}$ with dimension $r$, $\boldsymbol{d}=\boldsymbol{0}$, and $\boldsymbol{C}$ be a $r \times p$ matrix with 1 at the entries corresponding to the elements of $\boldsymbol{\beta}_j$ and 0 elsewhere. With this choice the Wald statistic $w$, defined in \eqref{wald}, takes the form 
    \begin{equation}\label{regrrel}
        w = \boldsymbol{\hat{\beta}}_j^T \mathbf{A}_j^{-1}\boldsymbol{\hat{\beta}}_j \,,
    \end{equation}
where $ \mathbf{A}_j$ is the submatrix of the asymptotic covariance matrix $\mathbf{A}=(\boldsymbol{X}^T\boldsymbol{\Tilde{W}}\boldsymbol{X})
        \ \hspace{-0.2cm} ^{\dagger}$ corresponding to the elements of $\boldsymbol{\beta}_j$. 
\end{proposition}
\begin{proof}
    Assuming that the prerequisites for $\boldsymbol{\beta}_j$ and $\boldsymbol{d}$, as specified in Proposition \ref{prop3.5}, are met, and with the matrix $\boldsymbol{C}$ taking on the following form:
    \begin{equation*}
     \text{r}~\left\{\begin{array}{@{}c@{}}\null\\\null\\\null\\\null\end{array}\right.
        \begin{pmatrix}
            0 & 1 & 0 & \hdots & 0 & 0 & \hdots & 0 \\
            0 & 0 & 1 & \hdots & 0 & 0 & \hdots & 0 \\
            \vdots &  & & \ddots & 0 & 0 & \hdots & 0 \\
             \undermatt{p}{0 & \undermat{r}{ 0 & 0 & \hdots & 1} & 0 & \hdots & 0} \\
        \end{pmatrix}
    \end{equation*}\vspace{0.65cm}\\
    Here, $\boldsymbol{\beta}_j$ represents the initial $r$ regression coefficients, given by $(\beta_1, \dots, \beta_r)^T$. Together with the definition \ref{defWald} there holds
    \begin{equation*}
    w \overset{\eqref{wald}}{=} (\boldsymbol{C}\hat{\boldsymbol{\beta}}-\boldsymbol{d})^T\left[\boldsymbol{C} \underbrace{(\boldsymbol{X}^T\boldsymbol{\Tilde{W}}\boldsymbol{X})
        \ \hspace{-0.2cm} ^{\dagger}}_{\boldsymbol{F^\dagger(\hat{\beta}|\mathbf{y})}}   \boldsymbol{C}^T\right]^{-1}(\boldsymbol{C}\hat{\boldsymbol{\beta}}-\boldsymbol{d}) \, = \boldsymbol{\hat{\beta}}_j^T \mathbf{A}_j^{-1}\boldsymbol{\hat{\beta}}_j \,,
    \end{equation*}
    which yields \eqref{regrrel}.
\end{proof}
In accordance with Proposition \ref{prop3.5}, the assessment of the relevance of a subvector $\boldsymbol{\beta}_j$ is determined by \eqref{regrrel}. If $\boldsymbol{\beta}_j$ is one-dimensional, the Wald statistic $w$ corresponds to the application of a t-test \cite{fahrmeir2013generalized}. The test statistic, denoted as $t_j$, quantifies the extent to which the estimated coefficient $\hat{\beta}_j$ deviates from $H_0$, taking into account the corresponding standard error, i.e.,
    \begin{equation*}
        t_j = \frac{\hat{\beta}_j}{\sqrt{a_{jj}}} \,,
    \end{equation*}
with $a_{jj}$ the $j$-th diagonal element of $\mathbf{A}=(\boldsymbol{X}^T\boldsymbol{\Tilde{W}}\boldsymbol{X})
        \ \hspace{-0.2cm} ^{\dagger}$. According to \cite{fahrmeir2013generalized}, $t_j$ is t-distributed with $n-p$ degrees of freedom and $H_0$ is rejected at significance level $\alpha$ if
    \begin{equation*}
        |t_j| > t_{1-\alpha/2}(n-p) \,.
    \end{equation*}
Alternatively, one can also perform the Likelihood-Ratio test using the likelihood ratio $\mathcal{L}$ statistic, defined as
    \begin{equation*}
        \mathcal{L} := -2\log\Big(L(\hat{\boldsymbol{\beta}}_{H_o}|\mathbf{y})/L(\hat{\boldsymbol{\beta}}|\mathbf{y})\Big) \,,
    \end{equation*}
where ${L(\hat{\boldsymbol{\beta}}|\mathbf{y})}$ is the likelihood function for the unrestricted estimator, and ${L(\hat{\boldsymbol{\beta}}_{H_o}|\mathbf{y})}$ is the likelihood function for the restricted estimator obtained by maximizing the likelihood subject to $H_0$. Analogous to the Wald statistic, $\mathcal{L}$ follows an asymptotic $\chi^2$-distribution with $r$ degrees of freedom under the null hypothesis $H_0$.

Linear hypothesis testing serves as a tool to assess the significance of estimated regression coefficients within a specified confidence level. This approach enables the determination of whether a particular predictor variable contributes meaningfully to the model's description or if a simpler model could suffice without sacrificing essential information. In contrast, model selection criteria, described in the next subsection, aim to identify the most suitable model for predicting outcomes accurately.

\begin{remark}{\textit{Note on Statistical Power and Variable Selection}}\\
In statistical inference, it is crucial to consider two types of errors. Type I error occurs when the null hypothesis is incorrectly rejected, mistakenly identifying an effect or relationship that does not exist. This risk is quantified by the significance level~$\alpha$. Conversely, Type II error arises when one fails to detect a genuine effect, incorrectly retaining the null hypothesis. This does not necessarily mean there is no effect; rather, it may reflect the test's limitations. Decisions to exclude terms from a model based solely on statistical significance should be made cautiously. While simple models are preferred for their ease of interpretation, overly strict criteria for variable selection may lead to important predictors being overlooked. Type II error risk is often denoted by $\beta$ (distinct from regression parameters). The probability that a statistical test will correctly reject a false null hypothesis is known as the power of the test and is represented by $1 - \beta$. High power increases confidence in hypothesis test outcomes, while low power raises doubts about non-significant findings. Statistical power relies on factors such as the significance level~$\alpha$, the sample size $n$, and the population effect size (ES) \cite{cohen1992statistical}.\end{remark}

In the context of predictive modeling for HVOF coating, domain expertise is essential in addressing statistical power challenges. Due to the constraints of a small sample size, compounded by the laborious and expensive nature of experiments (cf. Section \ref{sect_application}), the statistical power of hypothesis tests is inherently limited. Consequently, a thorough examination of regression coefficients was made in collaboration with thermal coating technicians to assess the relevance of predictors, particularly in cases where the performed test might not achieve statistical significance or the model selection criterion decides to exclude the respective effect. Further techniques for calculating and enhancing statistical power in regression analysis are explored in \cite{cohen2013statistical}.

% Subsubsection - Model Selection Criteria
\subsubsection{Model Selection Criteria}\label{AIC}

In practice, it is often necessary to compare different models and select the one which provides the best balance between model fit, reflecting the agreement with the observed data, and model complexity. Various criteria can be used for this purpose, including the Akaike Information Criterion (\acrshort{aic}) \cite{akaike1974new}. The \acrshort{aic} is based on the maximized log-likelihood function $\ell(\boldsymbol{\beta}|\mathbf{y})$ and is defined by:
    \begin{equation}\label{eq:AIC}
        \text{AIC} := -2\ell(\hat{\boldsymbol{\beta}}|\mathbf{y}) + 2p \;;
    \end{equation}
where $\hat{\boldsymbol{\beta}}$ is the maximum likelihood estimate of the model parameters, and $p$ is the number of parameters in the model. The \acrshort{aic} penalizes models with many parameters, thus favoring models that fit the data well but are not too complex. Smaller \acrshort{aic} values indicate better models, with a difference of 2 between \acrshort{aic} values suggesting strong evidence in favor of the model with the lower \acrshort{aic}. However, note that the \acrshort{aic} is a relative measure of model fit and should be used for comparing models within the same class. For example, the \acrshort{aic} cannot be used to compare a gamma regression model to a Poisson regression model, since they belong to different classes.

The application of model selection criteria such as the \acrshort{aic} is valuable in predicting \acrshort{hvof} coating properties based on process conditions. While it is important to develop accurate prediction models to optimize coating performance and ensure the desired coating properties, it is worth to consider that including too many irrelevant parameters in the model can introduce disturbances and adversely affect its predictive ability.

% % % % % % % % % % % % % % % % % % % % % % % %
% Section - Assessing Predicitive Performance %
% % % % % % % % % % % % % % % % % % % % % % % 5
\section{Assessing Predictive Performance of \acrshort{hvof} Coating Models}\label{sect_pred}

To assess the predictive performance of the \acrshort{hvof} regression model, the commonly employed technique of Leave-One-Out-Cross-Validation (\acrshort{loocv}) is utilized. It allows for a comprehensive evaluation of the model's generalization ability and its accuracy in forecasting coating properties. \acrshort{loocv} is particularly suitable for evaluating the model's generalization capability when only a limited number of observations is available \cite{wong2015performance}. The \acrshort{loocv} approach is a computationally intensive procedure, requiring the model to be fit $n$ times, i.e., once for each observation in the dataset. To improve computational efficiency, alternative resampling techniques such as k-fold cross-validation may be used.

The \acrshort{loocv} procedure involves iteratively fitting the model using all observations except one, and then using the fitted model to predict the response for the left-out observation. This is repeated for each observation in the dataset, resulting in $n$ predicted responses. The predicted response for the $i$-th observation is denoted as $\hat{y}^{(-i)}$, where the superscript $(-i)$ indicates that the $i$-th observation was left out during the fitting.

The prediction error for the $i$-th observation is defined as the difference between the predicted response and the observed response, i.e., $\epsilon_i = y_i - \hat{y}^{(-i)}$. 

\begin{definition}[ \cite{hastie2009elements}]
The \acrshort{loocv} estimate of the expected out-of-sample prediction error, i.e., the expected difference between the model's predictions and the true values of new, unseen observations, is defined by:
    \begin{equation*}
        CV_{(n)} := \frac{1}{n}\sum_{i=1}^{n}\epsilon_i^2 = \frac{1}{n}\sum_{i=1}^{n}(y_i - \hat{y}^{(-i)})^2 \;;
    \end{equation*}
where $n$ is the number of observations in the dataset.
\end{definition}

The \acrshort{loocv} estimate of the expected out-of-sample prediction error is an unbiased estimator of the true out-of-sample prediction error and can be used to compare the predictive performance of different models. The smaller the value of $CV_{(n)}$, the better the predictive performance of the model. In addition to the \acrshort{loocv}, we also use the $R^2$ statistic, which measures the proportion of variance in the observed response that is explained by the model. 

\begin{definition} 
The $R^2$ statistic is defined as:
    \begin{equation} \label{R2}
        R^2 := 1 - \frac{\sum_{i=1}^{n}(y_i - \hat{y}_i)^2}{\sum_{i=1}^{n}(y_i - \bar{y})^2} \,,
    \end{equation}
where $n$ is the number of observations, $y_i$ is the observed response for the $i$-th observation, $\hat{y_i}$ is the predicted response for the $i$-th observation, and $\bar{y}$ is the mean of the observed responses.
\end{definition}

The $R^2$ statistic can take values between 0 and 1, with higher values indicating a better fit of the model to the data. However, the $R^2$ statistic can be biased towards models with more predictors, even if the predictors have little or no effect on the response. To address this issue, the adjusted $R^2$ statistic is used, which adjusts $R^2$ for the number of predictors in the model. 

\begin{definition}
The adjusted $R^2$ statistic is defined as:
    \begin{equation}
        R_{adj}^2 :=  1 - \frac{(n-1)}{n - p}(1-R^2)\,,
    \end{equation}
where $p$ is the number of predictor variables in the model, $n$ is the number of observations, and $R^2$ is the statistic, defined in \eqref{R2}.
\end{definition}

The adjusted $R^2$ takes into account the trade-off between model complexity and model fit, and provides a more reliable measure of the model's predictive performance, compared to the traditional $R^2$, since it also accounts for the number of predictors $p$.

% % % % % % % % % % % % % 
% Section - DoE and CCD %
% % % % % % % % % % % % %
\section{Application to \acrshort{hvof} Coating: Practical Implementation} \label{sect_application}

The \acrshort{hvof} process is influenced by a multitude of variables, making it challenging to identify the most important factors that actually impact coating properties. In this study, a selection of five factors was deliberately chosen, guided by the knowledge of thermal spray experts who identified these variables as significant determinants influencing the \acrshort{hvof} process. Moreover, a well-designed experiment is crucial to efficiently collect data on the effects of various factors on the process outcomes. The selection of an optimal experimental design is essential within the domain of \acrshort{hvof} coating, primarily attributed to the considerable costs and time-intensive nature associated with conducting experiments using coating materials. Furthermore, a carefully planned experimental design enables strategic allocation of available experiments, maximizing information and providing valuable insights within a limited experimental scope. 

In industrial processes, statistical design of experiments (\acrshort{doe}) is considered a reliable technique for conducting experiments. A \acrshort{doe} allows for the systematic variation of process variables (= explanatory variables), which enables a more comprehensive understanding of their impact on the outcome. In contrast to the traditional one-factor-at-a-time approach, where interaction effects between two or more explanatory variables cannot be estimated, a \acrshort{doe} approach enables concise mathematical analysis of the resulting data and facilitates the identification of significant factors and their interactions. In addition, \acrshort{doe} allows researchers to investigate complex relationships between explanatory, revealing hidden insights, and supporting the optimization of industrial processes.

% Subsection - Central Composite Design
\subsection{Central Composite Design}

The central composite design (\acrshort{ccd}), a well-established and commonly employed experimental design in the field of industrial process optimization, is utilized in this work to acquire empirical data for the \acrshort{hvof} process. Compared to other designs, the \acrshort{ccd} is particularly useful as it can be efficiently used for fitting second-order models \cite{montgomery2012design}, i.e., estimating the effects of factors and their interactions in a quadratic form. The design incorporates a systematic variation of pre-defined factors, employing three levels $(-1,0,1)$ for each factor. Additional star points are included to enable the inclusion of quadratic terms in the model \cite{montgomery2012design}. 

\begin{figure}[h]
    \centering
    \includegraphics[width=0.9\textwidth]{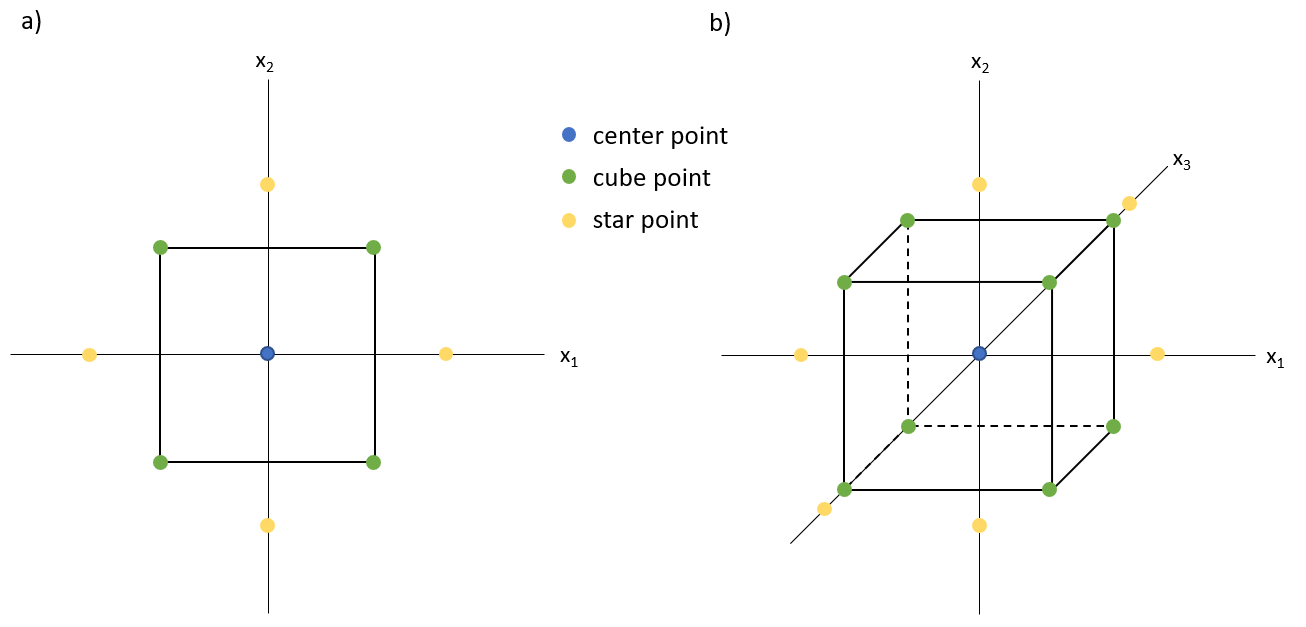}
    \caption{Central Composite Design with a) $k=2$ factors and b) $k=3$ factors.}
    \label{fig_ccd}
\end{figure}

Explanatory variables, which are given in quantitative form, are transformed into qualitative factors. The center point $x_0$, represented by the level $0$ for each factor, serves as a reference point and is used to assess the impact of factors on the system. The cube points correspond to the corners of the experimental region, represented by the levels $(-1,1)$. The star points are additional experimental points that are used to estimate the behavior beyond the linear response and to identify potential quadratic effects of factors. These points are positioned at a value of $\alpha$, where $\alpha$ is determined for a explanatory variable $x$ as
\begin{equation*}
        \alpha = x_0 \pm \delta_x\sqrt{k},
\end{equation*}
where $x_0$ is the center point, $\delta_x$ the difference between $x_0$ and the quantitative value that corresponds to $-1$, and $k$ is the number of explanatory variables under consideration. After transformation from quantitative values into qualitative ones, the values of $x_0$, $x_0\pm \delta_x$, and $\pm \alpha$ are replaced by $0$, $\pm 1$, and $\pm \sqrt{k}$ respectively. These qualitative values are then used in the design matrix $X$ to represent the explanatory variables (= factors). Figure~\ref{fig_ccd} depicts a Central Composite Design (\acrshort{ccd}) with $k=2$ and $k=3$ factors.

The number of factors $k$, which represent the process variables under investigation (cf. Section~\ref{ident}), directly determines the number of cube and star points in the \acrshort{ccd}. Specifically, there are $2^k$ cube points and $2k$ star points included in the design. Additionally, the \acrshort{ccd} consists of $n_c$ center points, where $n_c$ represents the total number of (potentially repeated) center points. To enhance the efficiency and accuracy of the design, a spherical \acrshort{ccd} was employed with the choice of $\alpha = \sqrt{k}$ concerning the star points. The spherical design allows for the estimation of effects of any factor with equal precision and reduces the risk of overemphasis on any factor. Thus, an optimal balance between precision and stability of the model parameters is obtained, which is important for receiving reliable estimates of the factor effects and their interactions. As recommended in \cite{montgomery2012design}, it is essential to randomize the experimental runs to avoid the influence of uncontrolled sources of variation.

% Subsection - Experimental Setup
\subsection{Experimental Setup}

% Subsubsection - Identification of influencing factors
\subsubsection{Identification of Influencing Factors}\label{ident}

Based on a review of the literature \cite{pukasiewicz2017influence, saaedi2010effects, zhao2004influence}, previous one-factor-at-a-time experiments, and expert knowledge by thermal spray coating experts of voestalpine TSM \cite{TSM}, five key factors, which are described in the next paragraph, were identified for systematic variation: powder feed rate (\acrshort{pfr}), stand off distance (\acrshort{sod}), lambda ($\lambda$), i.e. the stoichiometric ratio of oxygen to fuel, coating velocity (\acrshort{cv}), and total gas flow (\acrshort{tgf}). The schematic diagram in Figure~\ref{fig_ccd2} provides a comprehensive visual representation of the considered key process factors. Using these $k=5$ factors and conducting $n_c=7$ replications at the central point, a total of 49 trials were carried out, forming the \acrshort{ccd}. The experiments were conducted using a rotational setup that included a turning lathe, allowing for the application of the thermal spray coatings (cf. Figure~\ref{fig_real}).

\begin{figure}[h]
    \centering
    \includegraphics[width=0.9\textwidth]{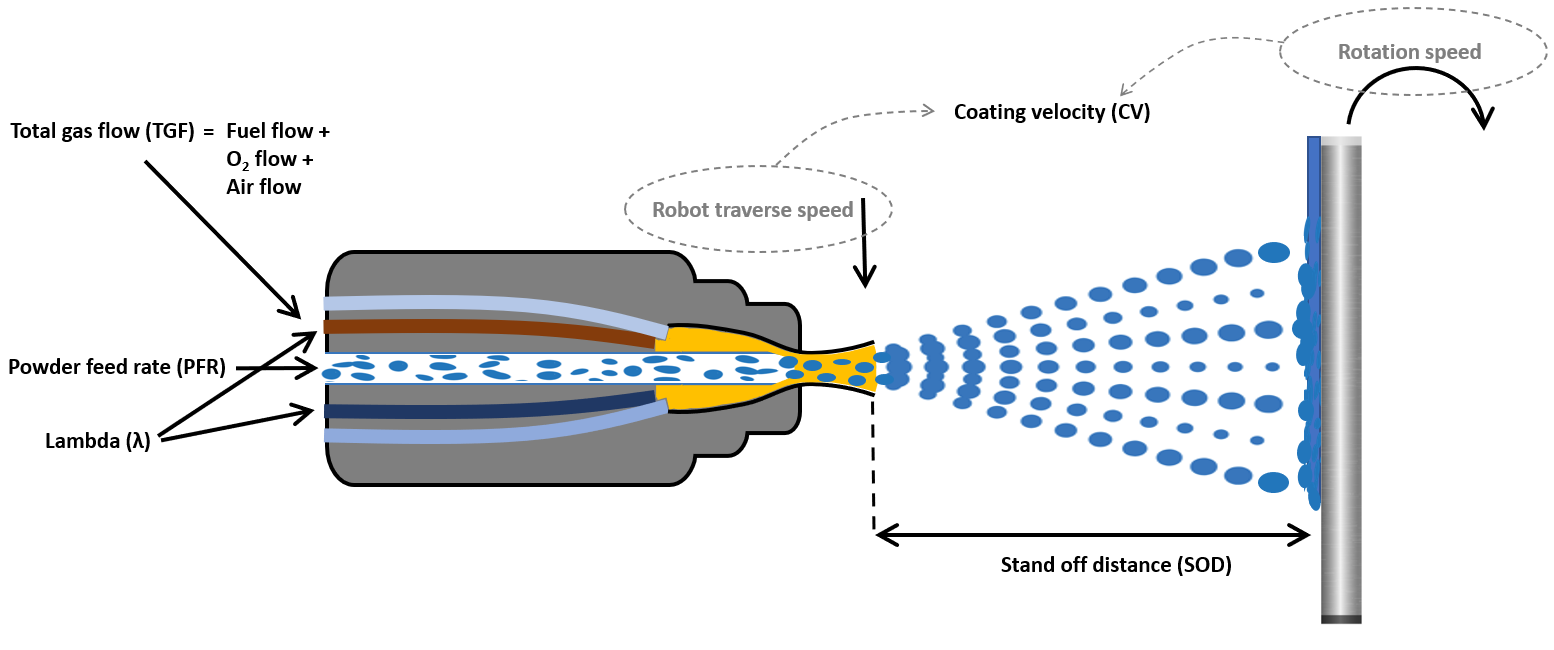}
    \caption{Illustration of the considered key factors in the HVOF coating process.}
    \label{fig_ccd2}
\end{figure}

The selected factors play a critical role in the \acrshort{hvof} coating process, exerting significant influence on the quality and performance of the resultant coatings. The \acrshort{pfr} governs the amount of coating material supplied, while the \acrshort{sod} regulates the spacing between the spray gun and the substrate. The stoichiometric ratio of oxygen to fuel ($\lambda$) ensures specific combustion conditions. Furthermore, the \acrshort{cv}, determined by the combined influence of the robot traverse speed and the rotational speed of the turning lathe (cf. Figure~\ref{fig_ccd2}), enables precise control over the deposition process. Finally, the \acrshort{tgf} is constituted by the summed gas flow of fuel, oxygen, and air, collectively governing the overall flow rate of the combustion gases.

Each of the five factors is accompanied by a designated set of predefined levels of variation, which are listed in Table~\ref{tab:CCD}. These levels were determined to cover a range of values that would effectively capture the variability and impact of these factors on the desired coating properties. The chosen levels allow for a systematic and comprehensive exploration of the parameter space. 
\begin{table}[h]
\centering
\begin{tabular}{l c c c c c}  
\toprule
\multicolumn{1}{c}{Factors} & \multicolumn{5}{c}{Coded values} \\
\cmidrule(r){2-6}
 & $-\sqrt{5}$ & $-1$ &$0$ &$1$ &$\sqrt{5}$\\
\midrule
Powder feed rate PFR [g/min] & 26.5 & 45 & 60 & 75 & 93.5 \\
         Stand off distance SOD [mm]  & 130 & 180 & 220 & 260 & 310 \\
         Lambda $\lambda$ & 0.72 & 0.84 & 0.94 & 1.04 & 1.16 \\
         Coating velocity CV [m/min] & 44 & 75 & 100 & 125 & 156 \\
         Total gas flow TGF [nl/min] & 531 & 615 & 683 & 751 & 835 \\
\bottomrule
\end{tabular}
    \caption{Levels of key factors for \acrshort{hvof} coating experiments depicted in Figure~\ref{fig_ccd}.}
    \label{tab:CCD}
\end{table}

% Subsubsection - Experimental Procedures
The \acrshort{hvof} coatings were produced using an Oerlikon Metco thermal spraying equipment, namely the DJ 2700 gas-fuel \acrshort{hvof} system with water-cooled gun assembly. The fuel gas used for these tests was propane, its amount and ratio defined by the two key factors \acrshort{tgf} and Lambda. For the process preparation, steel plates of type 1.4404 were welded onto an axis mounted on a turning lathe for rotational spraying. All samples were degreased with acetone and sandblasted with alumina before thermal spraying. The powder used for the spraying process was an agglomerated sintered tungsten carbide powder (WC-10Co-4Cr) with a grain size in the range of -45+15 µm, supplied by Oerlikon Metco. The photograph presented in Figure~\ref{fig_real} showcases the experimental setup employed during the \acrshort{hvof} coating process, wherein the dynamic engagement of the robot, turning lathe, and coating stream can be observed.

\begin{figure}[h]
    \centering
    \includegraphics[width=0.9\textwidth]{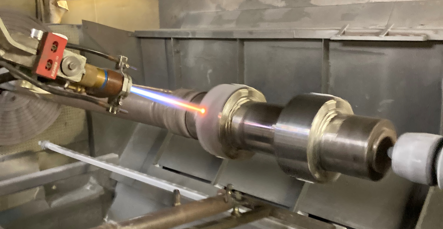}
    \caption{Photograph illustrating the experimental setup during the \acrshort{hvof} coating process, showing the robot, turning lathe, and coating stream in action.}
    \label{fig_real}
\end{figure}

Additional thermal spraying attributes like cooling, powder feed gas, pressure and number of passes, i.e., number of times the coating material was applied or sprayed onto the substrate during each experimental run, were kept constant throughout the experiments. In addition to the classic coating properties such as roughness, porosity, layer thickness, and surface hardness, the deposition rate, deposition efficiency, and in-flight particle properties such as particle velocity and particle temperature of the powder particles were measured.

Two different in-situ measurements were performed in the course of these trials. On the one hand, the in-situ particle characterization and on the other hand, the in-situ pyrometric temperature measurement. The particle characteristics were measured using a Spraywatch camera with the software SW4 (supplied by Oseir). The temperature of the sample surface was constantly measured using a Keller pyrometer. 

The surface roughness of the sprayed samples was measured using a mobil roughness tester Hommel Etamic Waveline W5. Coating hardness was assessed on the surface using a Cisam-Ernst S.r.l E-Computest mobile hardness tester, using a spheroconical diamond at a load of 5 kg and a testing time of 2 seconds. In addition to the surface characterization, cross-sections of each sample were prepared (according to internal preparation procedure WC) to analyse the coating thickness. The coating thickness and the coating porosity were determined using image analysis software, IMS Client, applied to microscopic images captured with a Zeiss Axio Observer.Z1m. Figure~\ref{fig_analytic} provides visual evidence of the observed variations in coating thickness, as captured in the microscopic images acquired from the IMS Client software.

\begin{figure}[h]
    \centering
    \includegraphics[width=0.9\textwidth]{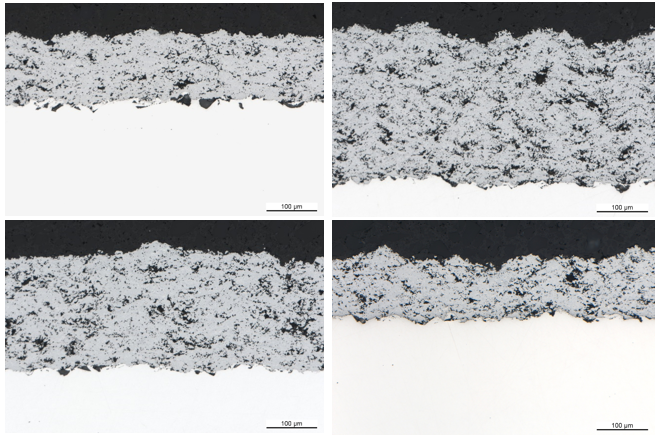}
    \caption{Microscopic images obtained from the IMS Client software, showing the observed variations in coating thickness across the cross-sectional profiles of the sprayed samples.}
    \label{fig_analytic}
\end{figure}

% % % % % % % % % % % % % % % % % % % % % % % %
% Section - Empirical Results of Experiments  %
% % % % % % % % % % % % % % % % % % % % % % % %
\section{Empirical Results of Experiments in \acrshort{hvof} Thermal Spraying}\label{sect_results}

In this section, the empirical findings derived from the comprehensive analysis of the experimental data are presented, demonstrating the effectiveness and utility of the gamma regression approach in analyzing the relationships between key process variables and coating properties. The analysis and modelling were performed using the statistical software R (version 4.2.2). The gamma regression models were implemented using the \textit{glm} function \cite{Rcore} with the Fisher scoring algorithm for estimating the regression coefficients $\hat{\boldsymbol{\beta}}$ (= ML estimates).

The properties listed in Table~\ref{tab:factors_all} serve as an overview and exemplify various potentially relevant properties of the HVOF process. This study focuses on analyzing 8 properties, which include in-flight characteristics (particle velocity and particle temperature), performance metrics (deposition rate and deposition efficiency), and coating attributes (thickness, roughness, hardness, and porosity). Analysis of the phase composition is typically not relevant for the coating material used. It should be emphasized, however, that the presented methodology holds equal relevance for the analysis of the other response variables in Table~\ref{tab:factors_all}.

Table~\ref{table:coefficients} provides a detailed illustrative example of the estimated regression coefficients and their corresponding standard errors of the deposition rate model and the deposition efficiency model. Additional tables containing analogous information for the remaining response variables can be found in the Appendix, specifically Table~\ref{table:coefficients1} and \ref{table:coefficients2}. These tables contain two kinds of models for each property, a full and a reduced version. The full model encompasses all predictor variables that can be estimated by utilizing the \acrshort{ccd} methodology, while the reduced model is derived through variable selection based on criteria such as the \acrshort{aic} and hypothesis testing for coefficient relevance, as described in Sections~\ref{testing} and \ref{AIC}. The model selection procedure involved the following steps: Initially, the full model was constructed, and non-significant coefficients were iteratively eliminated in a backward direction using the AIC as the guiding metric. If the removal of a coefficient resulted in a reduction in the AIC, the significance of the respective predictor was reassessed through a hypothesis test in consultation with thermal spray technicians. This consultation aimed to assess the practical significance of the coefficients in the context of thermal coating processes. Subsequently, a decision was made regarding the justification for the non-relevance of the coefficient, leading to its exclusion from the model. This approach was adopted due to observations indicating that the statistical power for many model coefficients fell below the recommended threshold of 0.8, as suggested by Cohen \cite{cohen2013statistical}.

Each row in Table~\ref{table:coefficients} corresponds to a specific predictor variable, i.e., main and quadratic effects of \acrshort{pfr}, \acrshort{sod}, Lambda, \acrshort{cv}, and \acrshort{tgf} and interaction effects between them. The associated coefficients (= \acrshort{ml} estimates) indicate the magnitude and direction of the predictor impact on the deposition rate and coating thickness. The values in parentheses next to the coefficients denote the respective standard errors. These regression coefficients and standard errors enable an assessment of the statistical significance of the associations between the predictor variable and the coating properties. The corresponding significance levels of the regression coefficients are denoted by asterisks. In particular, a significance level of $0.001$ is indicated by $^{***}$, $0.01$ by $^{**}$, $0.05$ by $^{*}$, and $0.1$ by $^{\bullet}$, where lower values (i.e., more stars) indicate a stronger level of statistical significance. Effects that do not exhibit significance symbols in the reduced model are considered to be of marginal relevance and have been incorporated into the analysis only due to their potential importance based on domain expertise.

\begin{table}[h]
\setlength{\tabcolsep}{2pt}
\small
\begin{center}
\begin{tabular}{l!{\hspace{0.3cm}} r!{}l!{\hspace{0.3cm}} r!{} l!{\hspace{0.3cm}} r!{} l!{\hspace{0.3cm}} r!{} l!{\hspace{0.3cm}}}
\toprule\toprule
\multicolumn{1}{c}{} & \multicolumn{4}{c}{Deposition rate} & \multicolumn{4}{c}{Deposition efficiency}\\
\cmidrule(r){2-5}
\cmidrule(r){6-9}
 & \multicolumn{2}{c}{\hspace{-0.5cm}full model} & \multicolumn{2}{c}{\hspace{-0.5cm}reduced model} & \multicolumn{2}{c}{\hspace{-0.5cm}full model} & \multicolumn{2}{c}{\hspace{-0.5cm}reduced model} \\
\midrule
Intercept      & \texttt{3.621} & \texttt{(0.024)}$^{***}$  & \texttt{3.634} & \texttt{(0.019)}$^{***}$   & \texttt{-0.474} & \texttt{(0.023)}$^{***}$  & \texttt{-0.455} & \texttt{(0.015)}$^{***}$  \\
PFR            & \texttt{0.268} & \texttt{(0.010)}$^{***}$  & \texttt{0.267} & \texttt{(0.009)}$^{***}$   & \texttt{0.006} & \texttt{(0.009)}         & \texttt{0.005} & \texttt{(0.009)}         \\
SOD            & \texttt{-0.020} & \texttt{(0.010)}$^{*}$   & \texttt{-0.021} & \texttt{(0.009)}$^{*}$    & \texttt{-0.020} & \texttt{(0.009)}$^{*}$    & \texttt{-0.021} & \texttt{(0.009)}$^{*}$    \\
Lambda         & \texttt{0.063} & \texttt{(0.010)}$^{***}$  & \texttt{0.064} & \texttt{(0.009)}$^{***}$   & \texttt{0.063} & \texttt{(0.009)}$^{***}$   & \texttt{0.064} & \texttt{(0.009)}$^{***}$   \\
CV             & \texttt{0.010} & \texttt{(0.011)}        &              &              & \texttt{0.010} & \texttt{(0.011)}         &                            \\
TGF            & \texttt{0.126} & \texttt{(0.010)}$^{***}$  & \texttt{0.126} & \texttt{(0.009)}$^{***}$   & \texttt{0.126} & \texttt{(0.009)}$^{***}$   & \texttt{0.126} & \texttt{(0.009)}$^{***}$   \\
PFR$^2$        & \texttt{-0.029} & \texttt{(0.009)}$^{**}$  & \texttt{-0.030} & \texttt{(0.009)}$^{**}$   & \texttt{0.008} & \texttt{(0.009)}         &                            \\
SOD$^2$        & \texttt{0.009} & \texttt{(0.010)}        &               &             & \texttt{0.008} & \texttt{(0.009)}         &                            \\
Lambda$^2$     & \texttt{-0.024} & \texttt{(0.010)}$^{*}$   & \texttt{-0.026} & \texttt{(0.009)}$^{**}$   & \texttt{-0.025} & \texttt{(0.009)}$^{*}$    & \texttt{-0.027} & \texttt{(0.009)}$^{**}$   \\
CV$^2$         & \texttt{-0.051} & \texttt{(0.010)}$^{***}$ & \texttt{-0.054} & \texttt{(0.009)}$^{***}$  & \texttt{-0.052} & \texttt{(0.009)}$^{***}$  & \texttt{-0.055} & \texttt{(0.008)}$^{***}$  \\
TGF$^2$        & \texttt{-0.051} & \texttt{(0.009)}$^{***}$ & \texttt{-0.052} & \texttt{(0.009)}$^{***}$  & \texttt{-0.052} & \texttt{(0.009)}$^{***}$  & \texttt{-0.054} & \texttt{(0.008)}$^{***}$  \\
PFR:SOD        & \texttt{0.013} & \texttt{(0.011)}        &               &             & \texttt{0.013} & \texttt{(0.011)}         &                            \\
PFR:Lambda     & \texttt{-0.006} & \texttt{(0.011)}       &               &             & \texttt{-0.006} & \texttt{(0.011)}        &                            \\
PFR:CV         & \texttt{0.013} & \texttt{(0.013)}        &               &             & \texttt{0.012} & \texttt{(0.012)}         &                            \\
PFR:TGF        & \texttt{0.018} & \texttt{(0.011)}        & \texttt{0.018} & \texttt{(0.010)}$^{\bullet}$ & \texttt{0.018} & \texttt{(0.011)}$^{\bullet}$ & \texttt{0.018} & \texttt{(0.010)}$^{\bullet}$ \\
SOD:Lambda     & \texttt{-0.005} & \texttt{(0.011)}       &                &            & \texttt{-0.005} & \texttt{(0.011)}        &                            \\
SOD:CV         & \texttt{0.003} & \texttt{(0.013)}        &                &            & \texttt{0.003} & \texttt{(0.012)}         &                            \\
SOD:TGF        & \texttt{0.008} & \texttt{(0.011)}        &                &            & \texttt{0.008} & \texttt{(0.011)}         &                            \\
Lambda:CV      & \texttt{-0.011} & \texttt{(0.013)}       &                &            & \texttt{-0.011} & \texttt{(0.012)}        &                            \\
Lambda:TGF     & \texttt{0.003} & \texttt{(0.011)}        &                 &           & \texttt{0.003} & \texttt{(0.011)}         &                            \\
CV:TGF         & \texttt{0.003} & \texttt{(0.013)}        &               &             & \texttt{0.003} & \texttt{(0.012)}         &                            \\
\hline 
AIC            & \multicolumn{2}{c}{\hspace{-0.5cm}\texttt{226.984}}                 & \multicolumn{2}{c}{\hspace{-0.5cm}\texttt{214.742}}                  & \multicolumn{2}{c}{\hspace{-0.5cm}\texttt{-174.236}}                 & \multicolumn{2}{c}{\hspace{-0.5cm}\texttt{-187.756}}                 \\
Log Likelihood & \multicolumn{2}{c}{\hspace{-0.5cm}\texttt{-91.492}}                 & \multicolumn{2}{c}{\hspace{-0.5cm}\texttt{-96.371}}                  & \multicolumn{2}{c}{\hspace{-0.5cm}\texttt{109.118}}                  & \multicolumn{2}{c}{\hspace{-0.5cm}\texttt{103.878}}                  \\
\hline\hline
\multicolumn{5}{l}{\scriptsize{$^{***}p<$\texttt{0.001}; $^{**}p<$\texttt{0.01}; $^{*}p<$\texttt{0.05}; $^{\bullet}p<$\texttt{0.1}}}
\end{tabular}
\caption{Estimated regression coefficients and standard errors (in brackets) for deposition rate and deposition efficiency models.}
\label{table:coefficients}
\end{center}
\end{table}

To evaluate the goodness-of-fit and performance of the models, the log-likelihood values play a crucial role. Specifically, the full models demonstrate higher log-likelihood values of $-91.492$ and $109.118$ compared to $-96.371$ and $103.878$ for the reduced models, indicating a stronger fit in capturing the observed data patterns compared to the reduced models. The reduced model exhibits lower \acrshort{aic} values of $214.742$ and $-187.756$ compared to the \acrshort{aic} values of $226.984$ and $-174.236$ obtained by the full model. These \acrshort{aic} values in Table~\ref{table:coefficients} indicate that the reduced model is favored over the full model in terms of achieving a better trade-off between model complexity and goodness-of-fit for both the deposition rate and deposition efficiency. Despite the full model potentially providing a better overall goodness-of-fit, the \acrshort{aic} criterion takes into account the complexity of the model and penalizes excessive complexity. 

Consistent with these findings, the supplementary Tables~\ref{table:coefficients1}, \ref{table:coefficients2}, and \ref{table:coefficients3} in the appendix uniformly show similar results regarding the \acrshort{aic} values and log-likelihoods. Notably, these results consistently favor the reduced models, indicating their ability to achieve a better balance between model complexity and goodness-of-fit. Moreover, across all regression models, each of the five explanatory variables demonstrates significant effects, providing robust evidence for their appropriate selection. Interestingly, the squared effects of individual factors exhibit greater statistical significance compared to the interaction effects. In addition, the results indicate that only the effects of Lambda and \acrshort{tgf} are consistently significant across all models, suggesting their shared dependence. This finding also highlights the intricate nature of the relationships involved. For instance, despite the expected correlation between deposition efficiency and deposition rate, it becomes apparent that these two properties cannot be adequately explained by the same set of parameters. This observation further emphasizes the technical challenges involved in handling and managing these interdependencies.

To ascertain whether the reduced model exhibits superior predictive performance compared to the full model, the metrics introduced in Section~\ref{sect_pred} are computed for each model individually. Table~\ref{tab:results} summarises the outcomes of the gamma regression analysis for in-flight properties (velocity and temperature), performance indicators (deposition rate and deposition efficiency), and coating properties (thickness, roughness, hardness, and porosity). Once again, the outcomes of both the full and reduced models are presented, highlighting their ability to model the studied properties. In addition to the number of coefficients $N_p$ and the model selection criterion \acrshort{aic} as in the preceding Table~\ref{table:coefficients}, this table also incorporates important performance metrics, namely $R^2$,$R^2_{adj}$, and $CV_{(n)}$, to measure the predictive quality of the regression models, as described in Section~\ref{sect_pred}.

\begin{table}[h]
\begin{tabular}{lllrrrrr}
  \hline
  \hline
   & Property & Model & \texttt{$N_p$}& AIC & \texttt{$R^2$} & \texttt{$R^2_{adj}$} &  \texttt{$CV_{(n)}$} \\ 
  \hline
   In-flight & velocity & full & \texttt{21} & \texttt{394.86} & \texttt{0.94} & \texttt{0.89} &  \texttt{288.2860} \\ 
     & velocity & reduced & \texttt{10} & \texttt{\textbf{375.58}} & \texttt{0.93} & \texttt{\textbf{0.92}} & \textbf{ \texttt{189.6886}} \\ 
     & temperature & full & \texttt{21} & \texttt{463.96} & \texttt{0.97} & \texttt{0.95} &  \texttt{1236.7553} \\ 
  & temperature & reduced & \texttt{10} & \texttt{\textbf{445.40}} & \texttt{0.97} & \texttt{\textbf{0.96}} &  \texttt{\textbf{947.5614}} \\
  \hline
  Performance  & deposition rate & full & \texttt{21} & \texttt{226.98} & \texttt{0.97} & \texttt{0.95} & \texttt{8.2556} \\ 
     & deposition rate & reduced & \texttt{10} & \texttt{\textbf{214.74}} & \texttt{0.97} & \texttt{\textbf{0.96}} & \texttt{\textbf{5.3105}} \\ 
  & deposition efficiency & full & \texttt{21} & \texttt{-174.24} & \texttt{0.91} & \texttt{0.85} & \texttt{0.0022} \\ 
    & deposition efficiency & reduced & \texttt{9} & \texttt{\textbf{-187.76}} & \texttt{0.89} & \texttt{\textbf{0.87}} & \texttt{\textbf{0.0014}} \\
    \hline
   Coating & thickness & full & \texttt{21} & \texttt{414.05} & \texttt{0.94} & \texttt{0.91} & \texttt{327.4603} \\ 
    & thickness & reduced & \texttt{10} & \texttt{\textbf{401.27}} & \texttt{0.93} & \texttt{\textbf{0.92}} & \texttt{\textbf{177.7381}} \\ 
    & roughness & full & \texttt{21} & \texttt{237.89} & \texttt{0.86} & \texttt{0.77} & \texttt{7.6912} \\ 
    & roughness & reduced & \texttt{11} & \texttt{\textbf{224.31}} & \texttt{0.85} & \texttt{\textbf{0.80}} & \texttt{\textbf{5.6868}} \\ 
    & hardness & full & \texttt{21} & \texttt{513.15} & \texttt{0.87} & \texttt{0.77} & \texttt{3044.5926} \\ 
    & hardness & reduced & \texttt{8} & \texttt{\textbf{499.42}} & \texttt{0.86} & \texttt{\textbf{0.82}} & \texttt{\textbf{1498.4070}} \\
    & porosity & full & \texttt{21} & \texttt{167.11} & \texttt{0.71} & \texttt{0.57} & \texttt{2.7142} \\ 
    & porosity & reduced & \texttt{13} & \texttt{\textbf{155.76}} & \texttt{0.69} & \texttt{\textbf{0.63}} & \texttt{\textbf{1.8773}} \\
   \hline
   \hline
\end{tabular}
\caption{Results of the gamma regression analysis for in-flight properties (velocity and temperature), performance indicators (deposition rate and deposition efficiency), and coating properties (thickness, roughness, hardness, and porosity)} using full and reduced models.
\label{tab:results}
\end{table}

Concerning the in-flight properties, both the full and reduced models demonstrate favorable results. The full model for particle velocity exhibits a high $R^2$ value of 0.94, indicating a strong fit to the observed data. However, taking into account the number of predictors $N_p$ in the model, it is advisable to consider the adjusted $R^2$ value of 0.89, which accounts for the model's complexity. Conversely, the reduced model for velocity yields a slightly lower $R^2$ value of 0.93, yet a higher adjusted $R^2$ value of 0.92 compared to the full model. These findings, coupled with lower values of the Akaike Information Criterion \acrshort{aic} and reduced out-of-sample prediction error $CV_{(n)}$, suggest that the reduced model offers superior predictive performance. Similar patterns emerge for the regression models investigating the other target variables in Table~\ref{tab:results}, consistently favoring the reduced models. According to the adjusted $R^2$, all reduced models demonstrate a good fit to the data, with values exceeding 0.8. However, the model for coating porosity yields less satisfactory results. This observation may be attributed to the volatile nature of the porosity measurement technique using Image Analysis, as discussed in \cite{ang2014review}

In addition to the findings in Table~\ref{tab:results}, Figure~\ref{fig_cv} provides a visual representation of the deposition rate predictions obtained from the full and reduced models. Each data point on the scatter plot represents an experimental trial, where the y-axis corresponds to the observed values, and the x-axis represents the \acrshort{loocv} predictions (refer to  Section~\ref{sect_pred}). The color-coded points differentiate between the center points, cube points, and star points obtained from the \acrshort{ccd}.

Upon analysis of the scatter plot, it is evident that the reduced model yields predictions that are closer to the diagonal line, indicating a higher degree of agreement between the predicted and observed values. This closer alignment implies a more accurate prediction of the deposition rate by the reduced model compared to the full model. Moreover, as expected, the prediction accuracy varies across different design points. The star points (yellow) demonstrate relatively lower predictive accuracy compared to the center points (blue) and the cube points (green), although this discrepancy is observed only in a subset of star points. This outcome underscores the challenges associated with extrapolating the model's behavior to regions outside the training data, emphasizing the need for caution when interpreting predictions for such points.

\begin{figure}[h]
    \centering
    \includegraphics[width=1\textwidth]{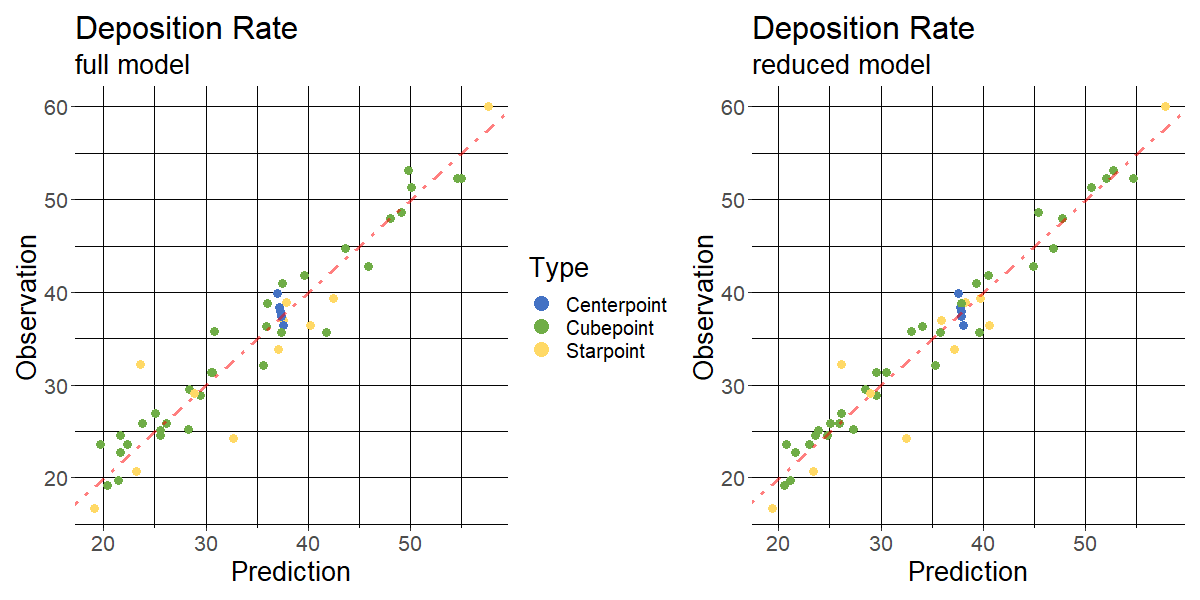}\vspace{-0.3cm}
    \caption{Scatter plot showing the comparison between observed values and \acrshort{loocv} predictions for deposition rate using the full and reduced models. The data points are color-coded based on the corresponding center points (blue), cube points (green), and star points (yellow) from the Central Composite Design.}
    \label{fig_cv}
\end{figure}

Overall, Figure~\ref{fig_cv} provides strong evidence supporting the superior predictive performance of the reduced model in estimating the deposition rate. The analysis of these visual results further strengthens the findings presented in Table~\ref{tab:results}, reinforcing the advantages of employing the reduced model in understanding and predicting the deposition rate more accurately. Furthermore, similar findings regarding the superior predictive performance of the reduced model are also observed for the other analyzed target variables.

% % % % % % % % % % % % %
% Section - Conclusion  %
% % % % % % % % % % % % %
\section{Conclusion}\label{sect_conclusion}

This study proposed a framework for modelling and predicting critical target variables in \acrshort{hvof} coating processes. By utilizing \acrshort{doe} and \acrshort{glm}s, accurate estimation of model parameters was achieved through maximum likelihood estimation. The framework incorporated a careful selection of predictor variables based on their significance and contribution to the coating properties, enhancing the model's interpretability and predictive performance. The application of this framework to experimental data from thermal spray coating experiments demonstrated its effectiveness in predicting target variables and providing insights into the relationships between factors and coating properties. The systematic variable selection process helps identify the most influential factors and eliminates irrelevant or redundant variables, simplifying the modelling process and improving the accuracy of predictions. The proposed framework has the potential to optimize thermal spray coating processes and contribute to the development of more efficient coating technologies in various industries. By developing a comprehensive understanding of the intricate interplay among process variables, material properties, and coating microstructure, manufacturers can enhance the functionality and performance of coated surfaces. This, in turn, can lead to improved product quality, extended component lifespan, and reduced maintenance costs.

In future investigations, we intend to expand our framework by introducing additional variables and exploring their interactions. This includes varying previously held constant factors, such as process gas pressures, across different levels to assess their impact. Furthermore, we plan to compare gamma regression models with alternative statistical models that require distinct distributional assumptions. While our previous work revealed the effectiveness of machine learning algorithms, particularly support vector machines, in predicting HVOF-related properties \cite{rannetbauer2024}, we aim to explore a hybrid approach combining these methodologies to further enhance predictive accuracy. Given the dynamic nature of the HVOF process, where process variables often deviate from target values, sensors will be installed within the booth to monitor these variations. Using advanced modeling techniques, this sensor data together with additional quantitative features are used to improve predictive capabilities. Further experiments will be conducted to ensure that a sufficient number of samples is available. Following satisfactory performance in predicting coating quality properties related to WCCoCr, the framework will be extended to other coating materials and their associated characteristics.

Overall, this study provides a systematic and data-driven approach to modeling and predicting coating properties in thermal spray coating. By leveraging this framework, researchers and practitioners can advance the understanding and optimization of thermal spray processes, leading to advancements in surface technology and its applications across industries. The variable selection process improves prediction accuracy and facilitates informed decision-making in the coating optimization process, contributing to the overall improvement of coating methodologies.

% % % % % % % % % % %
% Section - Support %
% % % % % % % % % % %
\section*{Declarations}

\subsection*{Acknowledgements}
The authors gratefully acknowledge voestalpine Stahl GmbH for their support through the research center, provision of materials, and financial contribution to this investigation. 

\subsection*{Funding}
SH and RR are also funded by the Austrian Science Fund (FWF): F6805-N36 within the SFB F68 ``Tomography Across the Scales''.

\subsection*{Availability of data and materials}
    The datasets generated and/or analyzed during the current study are not publicly available due to company confidentiality but are available from the corresponding author on reasonable request. The data used in this paper are proprietary and subject to confidentiality agreements. However, interested parties may request access to the data by contacting the corresponding author. Requests will be considered on a case-by-case basis, subject to company compliance and confidentiality agreements.

\subsection*{Competing interests}
 The authors declare that they have no competing interests.

\subsection*{Authors' contributions}
WR conceived and designed the study, gathered the data, performed the analysis, estimation, and modeling, and wrote the manuscript. The experiments were performed by WR and CH. CH contributed to data interpretation and provided critical revisions. SH supported with theoretical knowledge and provided critical revisions. All authors read and approved the final manuscript.
 
\clearpage

\printglossary[type=\acronymtype]
\clearpage
% % % % % % % % %
% Bibliography  %
% % % % % % % % %
\bibliographystyle{plain}
{\footnotesize
\bibliography{mybib}

\begin{thebibliography}{10}

\bibitem{TSM}
Technischer {S}ervice der voestalpine {S}tahl {G}mb{H}.
\newblock \url{https://www.voestalpine.com/technischerservice}.
\newblock Accessed: 2023-05-23.

\bibitem{akaike1974new}
H.~Akaike.
\newblock A new look at the statistical model identification.
\newblock {\em IEEE Transactions on Automatic Control}, 19(6):716--723, 1974.

\bibitem{ang2014review}
Andrew Siao~Ming Ang and Christopher~C Berndt.
\newblock A review of testing methods for thermal spray coatings.
\newblock {\em International Materials Reviews}, 59(4):179--223, 2014.

\bibitem{becker2021artificial}
Anderson Becker, Hip{\'o}lito~DC Fals, Angel~Sanchez Roca, Irene~BAF Siqueira,
  Felipe~R Caliari, Juliane~R da~Cruz, Rodolpho~F Vaz, Milton~J de~Sousa, and
  Anderson~GM Pukasiewicz.
\newblock Artificial neural networks applied to the analysis of performance and
  wear resistance of binary coatings {Cr$_3$C$_2$37WC18M} \ and \
  {WC20Cr$_3$C$_2$7Ni}.
\newblock {\em Wear}, 477:203797, 2021.

\bibitem{cohen1992statistical}
Jacob Cohen.
\newblock Statistical power analysis.
\newblock {\em Current directions in psychological science}, 1(3):98--101,
  1992.

\bibitem{cohen2013statistical}
Jacob Cohen.
\newblock {\em Statistical power analysis for the behavioral sciences}.
\newblock Academic press, 2013.

\bibitem{davis2004handbook}
Joseph~R Davis et~al.
\newblock {\em Handbook of {T}hermal {S}pray {T}echnology}.
\newblock ASM international, 2004.

\bibitem{dongmo2008analysis}
E~Dongmo, M~Wenzelburger, and R~Gadow.
\newblock Analysis and optimization of the {HVOF} process by combined
  experimental and numerical approaches.
\newblock {\em Surface and Coatings Technology}, 202(18):4470--4478, 2008.

\bibitem{fahrmeir1985consistency}
Ludwig Fahrmeir and Heinz Kaufmann.
\newblock Consistency and asymptotic normality of the maximum likelihood
  estimator in generalized linear models.
\newblock {\em The Annals of Statistics}, 13(1):342--368, 1985.

\bibitem{fahrmeir2013generalized}
Ludwig Fahrmeir, Thomas Kneib, Stefan Lang, Brian Marx, Ludwig Fahrmeir, Thomas
  Kneib, Stefan Lang, and Brian Marx.
\newblock {\em Regression: Models, {M}ethods and {A}pplications}.
\newblock Springer, 2013.

\bibitem{fauchais2014thermal}
Pierre~L Fauchais, Joachim~VR Heberlein, and Maher~I Boulos.
\newblock {\em Thermal {S}pray {F}undamentals: {F}rom {P}owder to {P}art}.
\newblock Springer Science \& Business Media, 2014.

\bibitem{gu2001computational}
S~Gu, CN~Eastwick, KA~Simmons, and DG~McCartney.
\newblock Computational fluid dynamic modeling of gas flow characteristics in a
  high-velocity oxy-fuel thermal spray system.
\newblock {\em Journal of thermal spray technology}, 10:461--469, 2001.

\bibitem{hardin2007generalized}
James~W Hardin and Joseph~M Hilbe.
\newblock {\em Generalized {L}inear {M}odels and {E}xtensions}.
\newblock Stata press, 2007.

\bibitem{hastie2009elements}
Trevor Hastie, Robert Tibshirani, Jerome~H Friedman, and Jerome~H Friedman.
\newblock {\em The {E}lements of {S}tatistical {L}earning: {D}ata {M}ining,
  {I}nference, and {P}rediction}.
\newblock Springer, 2009.

\bibitem{herman2000thermal}
Herbert Herman, Sanjay Sampath, and Robert McCune.
\newblock Thermal spray: current status and future trends.
\newblock {\em MRS bulletin}, 25(7):17--25, 2000.

\bibitem{ibrahim1991bayesian}
Joseph~G Ibrahim and Purushottam~W Laud.
\newblock On bayesian analysis of generalized linear models using jeffreys's
  prior.
\newblock {\em Journal of the American Statistical Association},
  86(416):981--986, 1991.

\bibitem{jalali2017fracture}
AZIZPOUR~M JALALI and M~Salehi.
\newblock Fracture {T}oughness of {HVOF} {T}hermally {S}prayed {WC-12Co}
  {C}oating in {O}ptimized {P}article {T}emperature.
\newblock {\em International Journal of advanced design and manufacturing
  technology}, 2017.

\bibitem{kuhnt2016residual}
Sonja Kuhnt, Andr{\'e} Rehage, Christina Becker-Emden, Wolfgang Tillmann, and
  Birger Hussong.
\newblock Residual {A}nalysis in {G}eneralized {F}unction-on-{S}calar
  {R}egression for an {HVOF} {S}praying {P}rocess.
\newblock {\em Quality and reliability engineering international},
  32(6):2139--2150, 2016.

\bibitem{li2009modeling}
Mingheng Li and Panagiotis~D Christofides.
\newblock Modeling and control of high-velocity oxygen-fuel {(HVOF)} thermal
  spray: {A} tutorial review.
\newblock {\em Journal of thermal spray technology}, 18:753--768, 2009.

\bibitem{liu2019prediction}
Meimei Liu, Zexin Yu, Yicha Zhang, Hongjian Wu, Hanlin Liao, and Sihao Deng.
\newblock Prediction and analysis of high velocity oxy fuel {(HVOF)} sprayed
  coating using artificial neural network.
\newblock {\em Surface and coatings technology}, 378:124988, 2019.

\bibitem{mojena2017neural}
Miguel Angel~Reyes Mojena, Angel~Sanchez Roca, Roberto~Sagaro Zamora,
  Mario~Sanchez Orozco, Hip{\'o}lito~Carvajal Fals, and Carlos Roberto~Camello
  Lima.
\newblock Neural network analysis for erosive wear of hard coatings deposited
  by thermal spray: {I}nfluence of microstructure and mechanical properties.
\newblock {\em Wear}, 376:557--565, 2017.

\bibitem{montgomery2012design}
D.C. Montgomery.
\newblock {\em Design and {A}nalysis of {E}xperiments, 8th Edition}.
\newblock John Wiley \& Sons, Incorporated, 2012.

\bibitem{nelder1972generalized}
John~Ashworth Nelder and Robert~WM Wedderburn.
\newblock Generalized linear models.
\newblock {\em Journal of the Royal Statistical Society Series A: Statistics in
  Society}, 135(3):370--384, 1972.

\bibitem{palanisamy2022effects}
Kalaiselvan Palanisamy, Srinu Gangolu, and Joseph~Mangalam Antony.
\newblock Effects of {HVOF} spray parameters on porosity and hardness of {316L}
  {SS} coated {M}g {AZ80} alloy.
\newblock {\em Surface and Coatings Technology}, 448:128898, 2022.

\bibitem{pan2016numerical}
Jiajing Pan, Shengsun Hu, Lijun Yang, Kunying Ding, and Baiqing Ma.
\newblock Numerical analysis of flame and particle behavior in an {HVOF}
  thermal spray process.
\newblock {\em Materials \& Design}, 96:370--376, 2016.

\bibitem{prasanna2018studies}
ND~Prasanna, C~Siddaraju, Gagan Shetty, MR~Ramesh, and Madhusudhan Reddy.
\newblock Studies on the role of {HVOF} coatings to combat erosion in turbine
  alloys.
\newblock {\em Materials Today: Proceedings}, 5(1):3130--3136, 2018.

\bibitem{pukasiewicz2017influence}
AGM Pukasiewicz, HE~De~Boer, GB~Sucharski, RF~Vaz, and LAJ Procopiak.
\newblock The influence of {HVOF} spraying parameters on the microstructure,
  residual stress and cavitation resistance of {FeMnCrSi} coatings.
\newblock {\em Surface and Coatings Technology}, 327:158--166, 2017.

\bibitem{Rcore}
{R Core Team}.
\newblock {\em {R}: {A} {L}anguage and {E}nvironment for {S}tatistical
  {C}omputing}.
\newblock R Foundation for Statistical Computing, Vienna, Austria, 2022.

\bibitem{rannetbauer2024}
Wolfgang Rannetbauer, Carina Hambrock, Simon Hubmer, and Ronny Ramlau.
\newblock Enhancing predictive quality in {HVOF} coating technology: A
  comparative analysis of machine learning lechniques.
\newblock {\em Procedia Computer Science}, 232:1377--1387, 2024.

\bibitem{ribu2022experimental}
Daniel~C Ribu, R~Rajesh, D~Thirumalaikumarasamy, Abdul~Razak Kaladgi, C~Ahamed
  Saleel, Kottakkaran~Sooppy Nisar, Saboor Shaik, and Asif Afzal.
\newblock Experimental investigation of erosion corrosion performance and
  slurry erosion mechanism of {HVOF} sprayed {WC-10Co} coatings using design of
  experiment approach.
\newblock {\em Journal of Materials Research and Technology}, 18:293--314,
  2022.

\bibitem{saaedi2010effects}
J~Saaedi, TW~Coyle, H~Arabi, S~Mirdamadi, and J~Mostaghimi.
\newblock Effects of {HVOF} process parameters on the properties of {Ni-Cr}
  coatings.
\newblock {\em Journal of thermal spray technology}, 19:521--530, 2010.

\bibitem{tabbara2011computational}
H~Tabbara, S~Gu, and DG~McCartney.
\newblock Computational modelling of titanium particles in warm spray.
\newblock {\em Computers \& fluids}, 44(1):358--368, 2011.

\bibitem{tillmann2022statistical}
Wolfgang Tillmann, Sonja Kuhnt, Ingor~Theodor Baumann, Arkadius Kalka,
  Eva-Christina Becker-Emden, and Alexander Brinkhoff.
\newblock Statistical {C}omparison of {P}rocessing {D}ifferent {P}owder
  {F}eedstock in an {HVOF} {T}hermal {S}pray {P}rocess.
\newblock {\em Journal of Thermal Spray Technology}, 31(5):1476--1489, 2022.

\bibitem{tyagi2021evaluation}
Ankit Tyagi, Qasim Murtaza, and RS~Walia.
\newblock Evaluation of the residual stress of {HVOF} sprayed carbon coating
  after wear testing conditions using {ANN} coupled {T}aguchi approach.
\newblock {\em Surface Topography: Metrology and Properties}, 9(3):035027,
  2021.

\bibitem{wedderburn1976existence}
Robert W~M Wedderburn.
\newblock On the existence and uniqueness of the maximum likelihood estimates
  for certain generalized linear models.
\newblock {\em Biometrika}, 63(1):27--32, 1976.

\bibitem{wong2015performance}
Tzu-Tsung Wong.
\newblock Performance evaluation of classification algorithms by k-fold and
  leave-one-out cross validation.
\newblock {\em Pattern Recognition}, 48(9):2839--2846, 2015.

\bibitem{zhang2009characterizations}
G~Zhang, A-F Kanta, W-Y Li, H~Liao, and C~Coddet.
\newblock Characterizations of {AMT-200} {HVOF} {NiCrAlY} coatings.
\newblock {\em Materials \& Design}, 30(3):622--627, 2009.

\bibitem{zhao2004influence}
Lidong Zhao, Matthias Maurer, Falko Fischer, Robert Dicks, and Erich
  Lugscheider.
\newblock Influence of spray parameters on the particle in-flight properties
  and the properties of {HVOF} coating of {WC-CoCr}.
\newblock {\em Wear}, 257(1-2):41--46, 2004.

\end{thebibliography}
}

% % % % % % %
% Appendix  %
% % % % % % %
\newpage
\appendix
\section{Supplementary Tables of Estimated Regression Coefficients and Standard Errors}

\begin{table}[h]
\setlength{\tabcolsep}{2pt}
\small
\begin{center}
\begin{tabular}{l!{\hspace{0.3cm}} r!{}l!{\hspace{0.3cm}} r!{} l!{\hspace{0.3cm}} r!{} l!{\hspace{0.3cm}} r!{} l!{\hspace{0.3cm}}}
\toprule\toprule
\multicolumn{1}{c}{} & \multicolumn{4}{c}{Particle velocity} & \multicolumn{4}{c}{Particle temperature}\\
\cmidrule(r){2-5}
\cmidrule(r){6-9}
 & \multicolumn{2}{c}{\hspace{-0.5cm}full model} & \multicolumn{2}{c}{\hspace{-0.5cm}reduced model} & \multicolumn{2}{c}{\hspace{-0.5cm}full model} & \multicolumn{2}{c}{\hspace{-0.5cm}reduced model} \\
\midrule
Intercept      & \texttt{6.126} & \texttt{(0.010)}$^{***}$    & \texttt{6.130} & \texttt{(0.005)}$^{***}$   & \texttt{7.452} & \texttt{(0.005)}$^{***}$    & \texttt{7.449} & \texttt{(0.003)}$^{***}$    \\
PFR            & \texttt{-0.006} & \texttt{(0.004)}         & \texttt{-0.006} & \texttt{(0.003)}        & \texttt{-0.007} & \texttt{(0.002)}$^{**}$    & \texttt{-0.007} & \texttt{(0.002)}$^{***}$   \\
SOD            & \texttt{0.050} & \texttt{(0.004)}$^{***}$    & \texttt{0.049} & \texttt{(0.003)}$^{***}$   & \texttt{-0.019} & \texttt{(0.002)}$^{***}$   & \texttt{-0.020} & \texttt{(0.002)}$^{***}$   \\
Lambda         & \texttt{-0.018} & \texttt{(0.004)}$^{***}$   & \texttt{-0.018} & \texttt{(0.003)}$^{***}$  & \texttt{0.025} & \texttt{(0.002)}$^{***}$    & \texttt{0.025} & \texttt{(0.002)}$^{***}$    \\
CV             & \texttt{-0.000} & \texttt{(0.005)}         &                  &          & \texttt{0.000} & \texttt{(0.002)}          &                             \\
TGF            & \texttt{0.048} & \texttt{(0.004)}$^{***}$    & \texttt{0.048} & \texttt{(0.003)}$^{***}$   & \texttt{0.051} & \texttt{(0.002)}$^{***}$    & \texttt{0.051} & \texttt{(0.002)}$^{***}$    \\
PFR$^2$        & \texttt{0.001} & \texttt{(0.004)}          &                   &         & \texttt{-0.001} & \texttt{(0.002)}         &                             \\
SOD$^2$        & \texttt{-0.030} & \texttt{(0.004)}$^{***}$   & \texttt{-0.030} & \texttt{(0.003)}$^{***}$  & \texttt{-0.001} & \texttt{(0.002)}         &                             \\
Lambda$^2$     & \texttt{0.000} & \texttt{(0.004)}          &                    &        & \texttt{-0.006} & \texttt{(0.002)}$^{*}$     & \texttt{-0.006} & \texttt{(0.002)}$^{**}$    \\
CV$^2$         & \texttt{0.002} & \texttt{(0.004)}          &                     &       & \texttt{-0.001} & \texttt{(0.002)}         &                             \\
TGF$^2$        & \texttt{-0.022} & \texttt{(0.004)}$^{***}$   & \texttt{-0.022} & \texttt{(0.003)}$^{***}$  & \texttt{-0.013} & \texttt{(0.002)}$^{***}$   & \texttt{-0.012} & \texttt{(0.002)}$^{***}$   \\
PFR:SOD        & \texttt{-0.009} & \texttt{(0.005)}$^{\bullet}$ & \texttt{-0.009} & \texttt{(0.004)}$^{*}$    & \texttt{-0.002} & \texttt{(0.002)}         &                             \\
PFR:Lambda     & \texttt{0.007} & \texttt{(0.005)}          & \texttt{0.007} & \texttt{(0.004)}$^{\bullet}$ & \texttt{0.002} & \texttt{(0.002)}          &                             \\
PFR:CV         & \texttt{0.000} & \texttt{(0.005)}          &                      &      & \texttt{0.000} & \texttt{(0.003)}          &                             \\
PFR:TGF        & \texttt{0.000} & \texttt{(0.005)}          &                       &     & \texttt{-0.004} & \texttt{(0.002)}$^{\bullet}$ & \texttt{-0.004} & \texttt{(0.002)}$^{\bullet}$ \\
SOD:Lambda     & \texttt{-0.005} & \texttt{(0.005)}         &                        &    & \texttt{0.004} & \texttt{(0.002)}          & \texttt{0.004} & \texttt{(0.002)}$^{\bullet}$  \\
SOD:CV         & \texttt{0.002} & \texttt{(0.005)}          &                         &   & \texttt{0.000} & \texttt{(0.003)}          &                             \\
SOD:TGF        & \texttt{0.014} & \texttt{(0.005)}$^{**}$     & \texttt{0.014} & \texttt{(0.004)}$^{***}$   & \texttt{0.003} & \texttt{(0.002)}          & \texttt{0.003} & \texttt{(0.002)}          \\
Lambda:CV      & \texttt{0.001} & \texttt{(0.005)}          &                   &         & \texttt{-0.001} & \texttt{(0.003)}         &                             \\
Lambda:TGF     & \texttt{0.001} & \texttt{(0.005)}          &                    &        & \texttt{-0.001} & \texttt{(0.002)}         &                             \\
CV:TGF         & \texttt{-0.000} & \texttt{(0.005)}         &                     &       & \texttt{-0.000} & \texttt{(0.003)}         &                             \\
\hline
AIC            & \multicolumn{2}{c}{\hspace{-0.5cm}\texttt{394.859}}        & \multicolumn{2}{c}{\hspace{-0.5cm}\texttt{375.575}}       & \multicolumn{2}{c}{\hspace{-0.5cm}\texttt{463.956}}        & \multicolumn{2}{c}{\hspace{-0.5cm}\texttt{445.404}}        \\
Log Likelihood & \multicolumn{2}{c}{\hspace{-0.5cm}\texttt{-175.429}}       & \multicolumn{2}{c}{\hspace{-0.5cm}\texttt{-176.788}}      & \multicolumn{2}{c}{\hspace{-0.5cm}\texttt{-209.978}}       & \multicolumn{2}{c}{\hspace{-0.5cm}\texttt{-211.702}}       \\
\hline\hline
\multicolumn{5}{l}{\scriptsize{$^{***}p<$\texttt{0.001}; $^{**}p<$\texttt{0.01}; $^{*}p<$\texttt{0.05}; $^{\bullet}p<$\texttt{0.1}}}
\end{tabular}
\caption{Estimated regression coefficients and standard errors (in brackets) for particle velocity and particle temperature models.}
\label{table:coefficients1}
\end{center}
\end{table}

\begin{table}[h]
\setlength{\tabcolsep}{2pt}
\small
\begin{center}
\begin{tabular}{l!{\hspace{0.3cm}} r!{}l!{\hspace{0.3cm}} r!{} l!{\hspace{0.3cm}} r!{} l!{\hspace{0.3cm}} r!{} l!{\hspace{0.3cm}}}
\toprule\toprule
\multicolumn{1}{c}{} & \multicolumn{4}{c}{Coating thickness} & \multicolumn{4}{c}{Coating roughness}\\
\cmidrule(r){2-5}
\cmidrule(r){6-9}
 & \multicolumn{2}{c}{\hspace{-0.5cm}full model} & \multicolumn{2}{c}{\hspace{-0.5cm}reduced model} & \multicolumn{2}{c}{\hspace{-0.5cm}full model} & \multicolumn{2}{c}{\hspace{-0.5cm}reduced model} \\
\midrule
Intercept      & \texttt{4.889} & \texttt{(0.040)}$^{***}$  & \texttt{4.893} & \texttt{(0.026)}$^{***}$  & \texttt{3.531} & \texttt{(0.026)}$^{***}$   & \texttt{3.524} & \texttt{(0.017)}$^{***}$   \\
PFR            & \texttt{0.227} & \texttt{(0.016)}$^{***}$  & \texttt{0.228} & \texttt{(0.015)}$^{***}$  & \texttt{0.021} & \texttt{(0.010)}$^{\bullet}$ & \texttt{0.023} & \texttt{(0.009)}$^{*}$     \\
SOD            & \texttt{0.007} & \texttt{(0.016)}        &              &             & \texttt{-0.008} & \texttt{(0.010)}        & \texttt{-0.007} & \texttt{(0.009)}        \\
Lambda         & \texttt{0.068} & \texttt{(0.016)}$^{***}$  & \texttt{0.066} & \texttt{(0.015)}$^{***}$  & \texttt{-0.033} & \texttt{(0.010)}$^{**}$   & \texttt{-0.034} & \texttt{(0.009)}$^{***}$  \\
CV             & \texttt{-0.266} & \texttt{(0.018)}$^{***}$ & \texttt{-0.266} & \texttt{(0.017)}$^{***}$ & \texttt{-0.042} & \texttt{(0.012)}$^{**}$   & \texttt{-0.042} & \texttt{(0.011)}$^{***}$  \\
TGF            & \texttt{0.038} & \texttt{(0.016)}$^{*}$    & \texttt{0.038} & \texttt{(0.015)}$^{*}$    & \texttt{-0.097} & \texttt{(0.010)}$^{***}$  & \texttt{-0.098} & \texttt{(0.009)}$^{***}$  \\
PFR$^2$        & \texttt{-0.008} & \texttt{(0.016)}       &                &           & \texttt{-0.007} & \texttt{(0.010)}        &                            \\
SOD$^2$        & \texttt{0.010} & \texttt{(0.016)}        &                &           & \texttt{0.032} & \texttt{(0.010)}$^{**}$    & \texttt{0.032} & \texttt{(0.009)}$^{**}$    \\
Lambda$^2$     & \texttt{-0.034} & \texttt{(0.016)}$^{*}$   & \texttt{-0.034} & \texttt{(0.015)}$^{*}$   & \texttt{0.016} & \texttt{(0.010)}         & \texttt{0.016} & \texttt{(0.009)}$^{\bullet}$ \\
CV$^2$         & \texttt{0.043} & \texttt{(0.016)}$^{*}$    & \texttt{0.043} & \texttt{(0.015)}$^{**}$   & \texttt{-0.023} & \texttt{(0.010)}$^{*}$    & \texttt{-0.022} & \texttt{(0.009)}$^{*}$    \\
TGF$^2$        & \texttt{-0.049} & \texttt{(0.016)}$^{**}$  & \texttt{-0.049} & \texttt{(0.014)}$^{**}$  & \texttt{0.001} & \texttt{(0.010)}         &                            \\
PFR:SOD        & \texttt{0.011} & \texttt{(0.018)}        &                &           & \texttt{-0.001} & \texttt{(0.012)}        &                            \\
PFR:Lambda     & \texttt{-0.025} & \texttt{(0.018)}       & \texttt{-0.025} & \texttt{(0.017)}       & \texttt{-0.001} & \texttt{(0.012)}        &                            \\
PFR:CV         & \texttt{-0.013} & \texttt{(0.021)}       &              &             & \texttt{-0.015} & \texttt{(0.014)}        &                            \\
PFR:TGF        & \texttt{0.022} & \texttt{(0.018)}        &             &              & \texttt{-0.001} & \texttt{(0.012)}        &                            \\
SOD:Lambda     & \texttt{-0.020} & \texttt{(0.018)}       &              &             & \texttt{0.012} & \texttt{(0.012)}         &                            \\
SOD:CV         & \texttt{0.018} & \texttt{(0.021)}        &             &              & \texttt{-0.012} & \texttt{(0.014)}        &                            \\
SOD:TGF        & \texttt{-0.006} & \texttt{(0.018)}       &              &             & \texttt{-0.024} & \texttt{(0.012)}$^{*}$    & \texttt{-0.024} & \texttt{(0.011)}$^{*}$    \\
Lambda:CV      & \texttt{0.016} & \texttt{(0.021)}        &             &              & \texttt{0.006} & \texttt{(0.014)}         &                            \\
Lambda:TGF     & \texttt{-0.015} & \texttt{(0.018)}       &              &             & \texttt{-0.066} & \texttt{(0.012)}$^{***}$  & \texttt{-0.067} & \texttt{(0.011)}$^{***}$  \\
CV:TGF         & \texttt{-0.033} & \texttt{(0.021)}       & \texttt{-0.033} & \texttt{(0.020)}       & \texttt{0.007} & \texttt{(0.014)}         &                            \\
\hline
AIC            & \multicolumn{2}{c}{\hspace{-0.5cm}\texttt{414.051}}                 & \multicolumn{2}{c}{\hspace{-0.5cm}\texttt{401.267}}                 & \multicolumn{2}{c}{\hspace{-0.5cm}\texttt{237.890}}                  & \multicolumn{2}{c}{\hspace{-0.5cm}\texttt{224.305}}                  \\
Log Likelihood & \multicolumn{2}{c}{\hspace{-0.5cm}\texttt{-185.026}}                & \multicolumn{2}{c}{\hspace{-0.5cm}\texttt{-189.633}}                & \multicolumn{2}{c}{\hspace{-0.5cm}\texttt{-96.945}}                  & \multicolumn{2}{c}{\hspace{-0.5cm}\texttt{-100.153}}                 \\
\hline\hline
\multicolumn{5}{l}{\scriptsize{$^{***}p<$\texttt{0.001}; $^{**}p<$\texttt{0.01}; $^{*}p<$\texttt{0.05}; $^{\bullet}p<$\texttt{0.1}}}
\end{tabular}
\caption{Estimated regression coefficients and standard errors (in brackets) for coating thickness and coating roughness models.}
\label{table:coefficients2}
\end{center}
\end{table}

\begin{table}[h]
\setlength{\tabcolsep}{2pt}
\small
\begin{center}
\begin{tabular}{l!{\hspace{0.3cm}} r!{}l!{\hspace{0.3cm}} r!{} l!{\hspace{0.3cm}} r!{} l!{\hspace{0.3cm}} r!{} l!{\hspace{0.3cm}}}
\toprule\toprule
\multicolumn{1}{c}{} & \multicolumn{4}{c}{Surface hardness} & \multicolumn{4}{c}{Coating porosity}\\
\cmidrule(r){2-5}
\cmidrule(r){6-9}
 & \multicolumn{2}{c}{\hspace{-0.5cm}full model} & \multicolumn{2}{c}{\hspace{-0.5cm}reduced model} & \multicolumn{2}{c}{\hspace{-0.5cm}full model} & \multicolumn{2}{c}{\hspace{-0.5cm}reduced model} \\
\midrule
Intercept      & \texttt{6.351} & \texttt{(0.026)}$^{***}$ & \texttt{6.352} & \texttt{(0.009)}$^{***}$    & \texttt{2.696} & \texttt{(0.028)}$^{***}$    & \texttt{2.706} & \texttt{(0.015)}$^{***}$    \\
PFR            & \texttt{-0.037} & \texttt{(0.010)}$^{**}$ & \texttt{-0.037} & \texttt{(0.010)}$^{***}$   & \texttt{0.005} & \texttt{(0.011)}          & \texttt{0.005} & \texttt{(0.010)}   \\
SOD            & \texttt{-0.036} & \texttt{(0.010)}$^{**}$ & \texttt{-0.034} & \texttt{(0.010)}$^{***}$   & \texttt{0.014} & \texttt{(0.011)}          & \texttt{0.015} & \texttt{(0.010)}  \\
Lambda         & \texttt{0.002} & \texttt{(0.010)}       & \texttt{0.003} & \texttt{(0.010)}          & \texttt{-0.029} & \texttt{(0.011)}$^{*}$     & \texttt{-0.029} & \texttt{(0.010)}$^{**}$          \\
CV             & \texttt{-0.019} & \texttt{(0.012)}      & \texttt{-0.019} & \texttt{(0.011)}$^{\bullet}$ & \texttt{0.008} & \texttt{(0.013)}          & \texttt{0.007} & \texttt{(0.012)} \\
TGF            & \texttt{0.119} & \texttt{(0.010)}$^{***}$ & \texttt{0.119} & \texttt{(0.010)}$^{***}$    & \texttt{-0.046} & \texttt{(0.011)}$^{***}$   & \texttt{-0.046} & \texttt{(0.010)}$^{***}$    \\
PFR$^2$        & \texttt{-0.011} & \texttt{(0.010)}      &              &               & \texttt{0.009} & \texttt{(0.011)}          &                             \\
SOD$^2$        & \texttt{0.003} & \texttt{(0.010)}       &           &                  & \texttt{0.004} & \texttt{(0.011)}          &                             \\
Lambda$^2$     & \texttt{0.008} & \texttt{(0.010)}       &            &                 & \texttt{-0.004} & \texttt{(0.011)}         &                             \\
CV$^2$         & \texttt{0.005} & \texttt{(0.010)}       &             &                & \texttt{0.038} & \texttt{(0.011)}$^{**}$     &  \texttt{0.036} & \texttt{(0.010)}$^{***}$                           \\
TGF$^2$        & \texttt{-0.003} & \texttt{(0.010)}      &              &               & \texttt{0.014} & \texttt{(0.011)}          &  \texttt{0.013} & \texttt{(0.010)}                           \\
PFR:SOD        & \texttt{0.013} & \texttt{(0.012)}       &               &              & \texttt{-0.014} & \texttt{(0.013)}         &                             \\
PFR:Lambda     & \texttt{-0.007} & \texttt{(0.012)}      &                &             & \texttt{0.024} & \texttt{(0.013)}$^{\bullet}$  & \texttt{0.024} & \texttt{(0.012)}$^{*}$                              \\
PFR:CV         & \texttt{0.004} & \texttt{(0.014)}       &                 &            & \texttt{0.029} & \texttt{(0.015)}$^{\bullet}$  & \texttt{0.029} & \texttt{(0.014)}$^{*}$                              \\
PFR:TGF        & \texttt{-0.012} & \texttt{(0.012)}      &                  &           & \texttt{-0.015} & \texttt{(0.013)}         &  \texttt{-0.015} & \texttt{(0.012)}                           \\
SOD:Lambda     & \texttt{-0.012} & \texttt{(0.012)}      &            &                 & \texttt{0.007} & \texttt{(0.013)}          &                             \\
SOD:CV         & \texttt{-0.014} & \texttt{(0.014)}      &             &                & \texttt{-0.008} & \texttt{(0.015)}         &                             \\
SOD:TGF        & \texttt{0.010} & \texttt{(0.012)}       &              &               & \texttt{-0.001} & \texttt{(0.013)}         &                             \\
Lambda:CV      & \texttt{-0.022} & \texttt{(0.014)}      & \texttt{-0.022} & \texttt{(0.013)}         & \texttt{-0.037} & \texttt{(0.015)}$^{*}$     & \texttt{-0.037} & \texttt{(0.014)}$^{*}$          \\
Lambda:TGF     & \texttt{0.025} & \texttt{(0.012)}$^{*}$   & \texttt{0.025} & \texttt{(0.011)}$^{*}$      & \texttt{-0.023} & \texttt{(0.013)}$^{\bullet}$ & \texttt{0.023} & \texttt{(0.012)}$^{\bullet}$      \\
CV:TGF         & \texttt{0.000} & \texttt{(0.014)}       &               &              & \texttt{-0.003} & \texttt{(0.015)}         &                             \\
\hline
AIC            & \multicolumn{2}{c}{\hspace{-0.5cm}\texttt{513.149}}                 & \multicolumn{2}{c}{\hspace{-0.5cm}\texttt{499.417}}                 & \multicolumn{2}{c}{\hspace{-0.5cm}\texttt{167.109}}                  & \multicolumn{2}{c}{\hspace{-0.5cm}\texttt{155.764}}                  \\
Log Likelihood & \multicolumn{2}{c}{\hspace{-0.5cm}\texttt{-234.574}}                & \multicolumn{2}{c}{\hspace{-0.5cm}\texttt{-240.708}}                & \multicolumn{2}{c}{\hspace{-0.5cm}\texttt{-61.554}}                  & \multicolumn{2}{c}{\hspace{-0.5cm}\texttt{-63.882}}                 \\
\hline\hline
\multicolumn{5}{l}{\scriptsize{$^{***}p<$\texttt{0.001}; $^{**}p<$\texttt{0.01}; $^{*}p<$\texttt{0.05}; $^{\bullet}p<$\texttt{0.1}}}
\end{tabular}
\caption{Estimated regression coefficients and standard errors (in brackets) for surface hardness and coating porosity models.}
\label{table:coefficients3}
\end{center}
\end{table}

\end{document}